\documentclass[11pt,A4paper]{article}

\usepackage{
   amsmath,
   amsfonts,
   amssymb,
   graphicx,
   epsfig,
   epstopdf,
   booktabs,
   algorithm,
   algorithmic,
   setspace,
   rotating,
 hyperref,
   subcaption,
   url,
   pbox,
   comment,
   multirow,
   enumitem,
   accents
   }
\usepackage{natbib}
\usepackage{circuitikz}

\newtheorem{theorem}{Theorem}

\newtheorem{observation}[theorem]{Observation}

        {\hspace*{\fill}$\Box$\par\vspace{4mm}}

\newenvironment{proof}{\noindent{\em Proof.}}%
        {\hspace*{\fill}$\Box$\par\vspace{4mm}}

\newcommand{\qed}{\hfill$\Box$\par\smallskip\noindent}
\newcommand{\ncoppa}[2]{\genfrac{}{}{0pt}{1}{#1}{#2}}

\def\S{\mathcal{S}}
\def\db{\Delta b}
\def\b{\tilde{b}}
\def\da{\Delta a}
\def\a{\tilde{a}}
\def\dA{\Delta A}
\def\Ae{\tilde{A}}
\def\u{\tilde{u}}
\def\comp{\beta} 
\def\compj{\comp_j} 
\def\glob{\mbox{\scriptsize sum}} 

\newcommand{\Real}{{\mathbb R}}
\def\jstar{i}

\def\gre{{\sc MaxMinGreedy}}

\def\sys{\mbox{\scriptsize\it SYS}}

\title{
Scheduling with Time‑Dependent Utilities: \\ Fairness and Efficiency
}

\author{Gaia Nicosia\thanks{Dipartimento di Ingegneria   Civile, Informatica e delle Tecnologie Aeronautiche,
Universit\`a degli Studi ``Roma Tre'', via della Vasca Navale 79,
00146 Rome, Italy, e-mail: gaia.nicosia@uniroma3.it}
\and  
Andrea Pacifici\thanks{Dipartimento di Ingegneria Civile e Ingegneria
Informatica, Universit\`a degli Studi di Roma ``Tor Vergata'', Via del Politecnico 1, 00133 Rome, Italy, e-mail:
andrea.pacifici@uniroma2.it}
\and
Ulrich Pferschy\thanks{Department of Operations and Information Systems, University of Graz, Universitaetsstrasse
15, 8010 Graz, Austria. e-mail: ulrich.pferschy@uni-graz.at} }

\begin{document}

\maketitle
\begin{abstract} 
A new class of multi-agent single-machine scheduling problems is introduced, where each job is associated with a self-interested agent with a utility function decreasing in completion time. We aim to achieve a fair solution by maximizing the minimum utility across all agents. We study the problem’s complexity and propose solution methods for several variants. 
For the general case, we present a binary search procedure to find the largest possible minimum utility, as well as an exact greedy-based alternative. 
Variants with release and due dates are analyzed, showing strong NP-hardness for arbitrary release dates, but weak NP-hardness for a single release-date job, and polynomial solvability when all jobs share processing times. 
For all these cases we also study the corresponding problem of finding efficient solutions where the sum of utilities is maximized.

We also examine settings where linear utility functions can be adjusted within budget constraints, exploring the impact on optimal schedules when intercepts or slopes are modified. 
From a single-agent perspective, we investigate the effect of improving one agent’s utility in the overall solution. 
Adding a new job to be inserted with the best possible utility gives rise to rescheduling problems, where different lower bounds depending on the utilities of the original fair schedule are imposed.
Finally, we consider a bi-level setting where a leader wants to enforce a certain target schedule by modifying utility functions while the follower computes a fair solution for the modified instance.
Our work contributes to scheduling theory, multi-agent systems, and algorithmic fairness, highlighting fairness-oriented objectives in competitive scheduling.
\end{abstract}

\section{Introduction}\label{sect:intro}

Scheduling problems have long been a central topic in operations research and computer science, with applications ranging from manufacturing and logistics to cloud computing and service operations. While classical models typically assume a single decision maker optimizing system-wide performance, increasing attention has been devoted in recent decades to multi-agent scheduling problems, where multiple autonomous agents compete for access to shared resources—most commonly, a single machine.

The study of multi-agent scheduling was initiated by the seminal works of  \citet{Baker2003} and \citet{bib:ampp2004}, which established the foundations for modeling scheduling problems involving multiple decision makers with potentially conflicting objectives. Since then, an extensive body of literature has examined settings in which agents control multiple jobs and seek to minimize individual completion-time-based cost functions {\cite[see, e.g., ][]  {AGNETIS2025,bib:acnp2019,Agnetis2015, HERMELIN2025,Decentralized-WELLMAN}.

More recently, growing attention has turned to the particular case where each job corresponds to a distinct agent \citep{bib:ntv2023}. This formulation captures scenarios with highly individualized preferences and has been primarily addressed through classical non-cooperative game-theoretic models, where agents select strategies such as machine choices in parallel-machine environments \cite[e.g., ][]{Correa2012,HOEKSMA2019}. However, strategic behavior alone does not account for collective welfare considerations, and may lead to highly unbalanced outcomes in which some agents incur disproportionately high costs.

In this work, we investigate a family of fairness-driven multi-agent single-machine scheduling problems in which each job corresponds to a self-interested agent with a utility function depending on its completion time. 
Rather than aggregating utilities through utilitarian or Nash-style objectives, we focus on the max-min fairness criterion, which seeks to maximize the minimum utility (or equivalently, minimize the maximum cost) across agents. This objective ensures that no agent is excessively penalized, promoting balanced and socially acceptable schedules, a relevant concern in competitive environments where fairness and equity play a central role.



 In many shared-resource scheduling settings, agents derive time-sensitive utility from the completion of their activities, with delays inducing a monotonic loss of value  \citep{Raut2008,Raut2008-2}. This naturally arises in economic environments where payoffs are discounted over time, so that later completion corresponds to a reduced net present value due to opportunity costs, uncertainty, or foregone alternative uses of capital.
 When access to a common resource is constrained, objectives focusing solely on aggregate performance may result in highly uneven completion times, disproportionately penalizing some agents. So, a max–min utility criterion provides a principled mechanism to mitigate such imbalances by explicitly limiting the extent to which any individual participant can be disadvantaged by scheduling decisions.
 
 Similar considerations apply in market contexts characterized by rapidly decaying demand or willingness to pay. For products with short commercial lifecycles or strong initial hype, revenues associated with delayed production or delivery may decrease sharply over time~\citep{HypeCycle}. In such cases, maximizing total utility can still lead to schedules in which some agents experience delays that render their participation economically marginal. Optimizing the minimum utility instead promotes a more equitable allocation of completion times, ensuring that all agents retain a non-negligible share of the attainable value while preserving sensitivity to time-dependent economic effects.

We examine several variants of the problem, taking into account different agent‑specific utility functions and temporal constraints, such as release dates and due dates. We also analyze versions of the problem in which certain input parameters can be adjusted either to obtain specific classes of solutions or to analyze how such modifications affect the resulting solutions.
Our main contributions are as follows:

\begin{itemize}[noitemsep,nolistsep]
    \item We propose exact solution approaches, including a binary search framework for identifying the largest achievable minimum utility and a greedy alternative with wide applicability.
    \item We establish complexity results for these variants, showing, for example, strong NP-hardness in the presence of arbitrary release dates and weak NP-hardness when only a single job has a release constraint, while identifying polynomially solvable situations such as the equal processing times case.
    \item The framework is extended to scenarios where linear utility functions can be adjusted within budget constraints. We analyze how modifications to intercepts or slopes influence equilibrium schedules.
    \item From a single-agent perspective, we study how improving one agent’s utility affects the overall fairness-oriented solution, providing structural insights into the trade-off between individual incentives and collective guarantees.
    \item We also investigate our problem in a leader-follower setting, where the leader aims to enforce a target solution by a adjusting the utility functions of the follower-agents who, in turn, will have their jobs arranged according to a fair schedule.
\end{itemize}
The individual problems treated in this paper will be introduced in Section~\ref{sect:problem} which also describes the organization of the paper.

\section{Related literature}\label{sect:literature}



This section reviews the existing literature on scheduling problems that are closely related to the focus of this study, aiming to position our work within the broader research context. We begin with multi-agent scheduling, examining Pareto-optimal solutions and fairness-related considerations, which remain central in practical applications, then turn to scenarios where jobs themselves are treated as agents, and finally consider scheduling problems in which job values evolve over time, emphasizing the challenges posed by dynamic priorities and the opportunities they create for adaptive approaches.

Multi-agent scheduling involves multiple agents, each with distinct objectives and job sets, competing for shared processing resources. This paradigm is prevalent in various applications, including manufacturing systems, cloud computing, and collaborative robotics, where efficient resource allocation among competing entities is crucial.

Initial research predominantly focused on scenarios involving two agents \cite[see, e.g. ][]{bib:ampp2004,Baker2003}. However, real-world applications  frequently encompass environments with more than two agents, necessitating more sophisticated scheduling strategies capable of handling higher-dimensional interactions among heterogeneous objectives. This has led to a substantial expansion of the literature on multi-agent scheduling beyond the two-agent setting.
For example, \cite{agnetis2007multi} analyzed multi-agent single-machine scheduling where each agent owns a set of non-preemptive jobs, highlighting key complexity results across several classical performance criteria. Similarly, \cite{yuan2017single} studied single-machine models with a fixed number of competing agents, identifying problem variants that remain polynomially solvable and offering constructive algorithmic solutions. In the same vein,  \cite{li2020competitive} examined competitive multi-agent scheduling under total weighted late work objectives, proposing methods for determining Pareto-optimal solutions that reflect trade-offs among agents.
More recently, \cite{Wang2021} investigated a multitasking scheduling model on a cloud manufacturing platform involving multiple competing agents, providing updated complexity results and relevant solution procedures, while  \cite{bib:CHEN2026_equallenghtmultiagent} addressed several single-machine feasibility problems with equal-length jobs considering different agent objectives.

Fairness has become an increasingly important theme in this context, as solutions driven purely by efficiency may disproportionately benefit certain agents \citep{bib:acnp2019}. Consequently, incorporating fairness into scheduling algorithms has become essential in systems where multiple decision makers compete for shared resources. In recent years, a growing body of literature has addressed fairness from a variety of perspectives \cite[see, e.g., ][for recent contributions]{AGNETIS2025,HERMELIN2025}.
In particular,  \cite{AGNETIS2025} investigate how the introduction of fairness constraints reshapes the structure of multi-agent single-machine scheduling problems and affects the trade-off between efficiency and equity. Complementarily, \cite{HERMELIN2025} study fairness-aware scheduling models, showing that equity requirements can give rise to new algorithmic challenges and significantly modify the set of feasible solutions.




A social welfare solution (i.e., one maximizing the overall utility of the system) may be highly unbalanced in terms of individual utilities and hence unsatisfactory for some of the agents. 
On the other hand, the global utility of a fair schedule can be potentially far away from its maximum value.
A natural question then arises to assess what is the maximum loss in terms of social welfare when scheduling jobs according to a fair schedule. 
In this regard, a relevant fraction of the literature deals with finding upper bounds to the so-called \emph{Price of Fairness} (PoF), introduced by~\cite{bib:bft2011} and then considered in several papers \citep{bib:acnp2019,bib:npp2017, bib:zzl2020}   as the relative loss of a fair solution with respect to the utility of a total utility maximizer, in the context of general resource allocation problems.
A recent contribution establishes tight bounds on the PoF under envy-freeness for time-dependent utility models, offering insights into how structural parameters affect fairness-induced utility loss \citep{Vardi2024} .

In related research streams, jobs may themselves be treated as agents. In cooperative scheduling games, for example, players may exchange utility to stabilize a globally optimal schedule \citep{CURIEL1989}, whereas in non-cooperative environments the focus shifts to the interplay between incentives and fairness constraints, as illustrated by \citet{Vardi2024}. 

A final set of contributions concerns scheduling problems in which job values evolve over time. 
Deteriorating job values are considered in the work of \cite{Raut2008}, who analyze a single-machine problem where the value of each job decreases as its starting time is delayed. Their study considers unrestricted, truncated, and capacity-constrained deterioration functions, showing that while some special cases---such as linear unrestricted deterioration---are polynomially solvable, the general problem is unary NP-hard. They model the problem via a time-indexed formulation and propose several heuristics whose empirical performance is carefully assessed. 
Building on this line of research, 
 \cite{Raut2008-2} extend the model to a capacitated setting, where the machine can process only a limited number of jobs simultaneously. They demonstrate that the added capacity constraints preserve the NP-hard nature of the problem and develop new heuristics based on multiplicative piecewise metrics that perform favorably relative to existing methods. 
These works highlight how time-dependent job utilities introduce significant structural complexity and require tailored algorithmic strategies. 

\section{Problem statement}\label{sect:problem}


Given a set of $n$ non-preemptive jobs $J$, we consider a class of single-machine scheduling problems in which each job is associated to an agent. For any $j\in J$, a processing time $p_j$, a single-valued non-increasing \emph{utility function} $u_j(\cdot)$ 
and, possibly, a release date $r_j$ and a due date $d_j$ are given.
Any schedule $\sigma$ of the jobs in $J$ implies a completion time $C_j(\sigma)$ for every job/agent $j$.
We define the utility of agent $j$ in schedule $\sigma$ as $u_j(C_j(\sigma))$. Depending on the context, we will write equivalently $u_j(C_j)$ or $u_j(\sigma)$ or, simply, $u_j$.

With no loss of generality, we assume that $u_j \geq 0$ for all feasible schedules. This can be easily reached by adding a suitably large constant to all utility functions, namely $M:=\max\{0, - u_j(\bar{C})\}$ where
$\bar{C}$ is some upper bound on all completion times, e.g., $P:=\sum_j p_j$, or $\max_j r_j + \sum_j p_j$, if release dates $r_j$ are present.
We also sometimes  use of the following generalized inverse utility function \citep{bib:general_inverse} giving the latest possible completion for reaching a utility target value $u^T$:
\begin{equation}\label{eq:definverse}
u^{-1}_j(u^T) := \sup\{ C \in\Real \mid  u_j(C) \geq u^T\}.
\end{equation}
Hereafter, we  deal with the cases of general non-increasing linear utility functions but sometimes
also restrict our attention to linear utility functions $u_j(C_j)=b_j - a_j C_j$ with $a_j \geq 0$.

In the remainder of this work, we indicate by $\S$ the set of feasible semi-active schedules (in which each job cannot start earlier without violating the given order) for the problem at hand.
Table~\ref{tab:notation} reports the notation adopted for all the problems addressed in the remainder of the paper.

\medskip
The fair scheduling problems considered in this work aim at determining a \emph{fair schedule} $\sigma_F$, namely a schedule that maximizes the minimum utility among all agents.
Let $u_{\min}(\sigma) := \min_{j \in J} u_j(\sigma)$, then the fair schedule is defined as
\begin{equation}\label{eq:fairsoldef}
    \sigma_F \in \arg\max_{\sigma \in \mathcal{S}} \left\{ u_{\min}(\sigma) \right\},
    \qquad
    u_F := u_{\min}(\sigma_F).
\end{equation}
As a notable special case, when all agents share the same utility function, computing $\sigma_F$ reduces to minimizing the maximum completion time over all jobs, that is, the classical makespan objective $C_{\max}$.

For a given problem, in addition to a fair solution, we also consider a \emph{system-optimal solution}, that is, a schedule $\sigma_{\sys}$ maximizing the global utility (or social welfare) of the system.
Let $u_{\glob}(\sigma) := \sum_{j \in J} u_j(\sigma)$; then the system-optimal schedule is defined as
\begin{equation}\label{eq:syssol}
    \sigma_{\sys} \in \arg\max_{\sigma \in \mathcal{S}} \left\{ u_{\glob}(\sigma) \right\},
    \qquad
    u_{\sys} := u_{\glob}(\sigma_{\sys}).
\end{equation}
%
\begin{table}[htb]
    \centering
{\small    
\begin{tabular}{|l|l|}
\toprule
$J:=\{1,\ldots,n\}$ & set of jobs/agents indexed by $j$ \\
$p_j$, $r_j$, $d_j$, & processing time, release date, due date of job $j$\\
$P:= \sum_j^n p_j$ & total processing time of the jobs \\
$s_j(\sigma)$, $s_j$ & starting time of job $j$ in schedule $\sigma$\\
$C_j(\sigma)$, $C_j$ & completion time of job $j$ in schedule $\sigma$ 
\\
$u_j(C_j(\sigma))$, $u_j(C_j)$, $u_j(\sigma)$  & utility of job $j$ completing at $C_j$ in schedule $\sigma$ \\
$u_j(C_j) = b_j - a_j C_j$ & linear utility of job $j$ completing at $C_j$\\
$- a_j$ & slope of the linear function \\
$b_j$ & intercept of the linear function \\
$\S$ & set of feasible (semi-active) schedules \\
$u_{\min}(\sigma) = \min_{j\in J}\{u_j(\sigma)\}$ & minimum utility attained in schedule $\sigma$ \\
$u_{\glob}(\sigma) = \sum_{j\in J} u_j(\sigma)$ & total utility of schedule $\sigma$ \\
$\sigma_{F} = \arg\max_{\sigma\in\S}\{u_{\min}(\sigma)\}$ & fair (max--min) solution \\
$\sigma_{\sys} = \arg\max_{\sigma\in\S}\{u_{\glob}(\sigma)\}$ & system-optimal solution \\
$u_{F} = u_{\min}(\sigma_{F})$ & fair (max--min) solution value \\
$u_{\sys} = u_{\glob}(\sigma_{\sys})$ & system-optimal solution value \\
$\Delta a_j$ $\Delta b_j$ & var \\
$\delta_j$ & disruption \\
\bottomrule
\end{tabular}
}
  \caption{Notation}
    \label{tab:notation}
\end{table}
Depending on the characteristics of the agents' utility functions, finding a system-optimal schedule $\sigma_{\sys}$ may amount to minimizing a regular objective function, for which many results are available in the scheduling literature.
In particular, in the case of linear utility functions, since
\begin{equation}\label{eq:sum-utili}
\max \sum_{j \in J} u_j(C_j)
= \max \sum_{j \in J} (b_j - a_j C_j)
= \sum_{j \in J} b_j - \min \sum_{j \in J} a_j C_j,
\end{equation}
computing $\sigma_{\sys}$ is equivalent to minimizing the weighted sum of completion times.
This problem has been extensively studied.
Indeed, in the absence of additional constraints or special system features, when $u_j(C_j) = b_j - a_j C_j$, a system-optimal solution can be obtained by sequencing jobs in non-increasing order of $a_j / p_j$, according to Smith’s rule (also known as the WSPT rule).

Obviously, in the general case of non-increasing $u_j(\cdot)$, finding a schedule that maximizes the sum of utilities depends on the specific form of the utility functions and may be far from trivial.  
For instance, suppose the utility of agent $j$ is given by  
$
u_j(C_j) = \min\{d,\, b_j - a_j C_j\}.
$  
In this case, the utility of agent $j$ starts decreasing from $d$ whenever $j$ completes after the ``due date'' $(b_j - d)/a_j$. This formulation encompasses the classical scheduling problem $1||\sum w_j T_j$ (minimization of total weighted tardiness), which is well known to be strongly NP-hard (cf.\ p.~359 in~\cite{bib:lenstra1977}). Finding a system optimal solution remains strongly NP-hard even when $u_j = b_j - a_j f(C_j)$ and $f(\cdot)$ is a piecewise linear function \cite[cf. Theorem 5 in][] {bib:SmithNonLinear}.  
Interestingly, the quadratic case $u_j = b_j - a_jC_j^2$ becomes polynomially solvable when $a_j$ is identical for all agents, since the SPT rule yields an optimal schedule to the standard scheduling problem $1||\sum C_j^2$  \citep{townsend78}. However, the complexity status remains open for arbitrary agent weights in the case $u_j(C_j) = b_j - a_j C_j^2$~\citep{ChenPottsWoeginger1998,bib:SmithNonLinear}.
To summarize we can state the following:
\begin{observation}\label{obs:maxsum}
    For general (resp., linear) non-increasing utility functions of the agents, maximizing the total utility, i.e., finding a system optimal solution is strongly NP-hard (resp., polynomially solvable through the Smith's rule).
\end{observation}

It is worth to mention here that the global utility of a fair solution can be quite small compared to that of a system optimum and, concurrently, that the fairness value $u_{\min}(\sigma_{\sys}) = \min_j \{u_j(\sigma_{\sys})\}$ (i.e., the utility of the worse-off agent in a system optimal solution) can be arbitrarily small compared to that of a fair solution
$u_F = \min_j \{u_j(\sigma_F)\}$.
This can be shown through a simple example in which the ratio between 
$
u_{\glob}(\sigma_{\sys}) = \sum_{j \in J}u_j(\sigma_{\sys})
$
and $ u_{\glob}(\sigma_F) = \sum_{j\in J} u_j(\sigma_F)$
becomes arbitrarily large, even when all jobs have identical processing times.
Consider $n=2$ jobs with $p_1=p_2=1$ and linear utility functions $u_1(C_1)=2M+1-MC_1$ and $u_2(C_2)=2-C_2$. If job $1$ is scheduled first, we obtain $u_1=M+1$ and $u_2=0$, which yields the global optimal schedule $\sigma_{\sys}$. Conversely, scheduling job $2$ first results in $u_2=1$ and $u_1=1$, making this ordering the fair schedule $\sigma_F$.
Due to this negative result, in this study we do not address Price of Fairness issues any further.

In this work we consider four different classes of problems depending on the agents/jobs perspective.  These problems are rigorously  stated in the following Sections~\ref{subsec:fixed}--\ref{subsec:stackel}.
As it is customary in scheduling field, 
we denote the addressed problems using a three-field notation. The fair scheduling problem defined above is then 
$1|\beta|u_{\min}$, where additional peculiarities are specified in the middle field $\beta$. The corresponding system-optimum is instead found by solving problem $1|\beta|\sum u_j$.
We give an synthetic description of these problems and the corresponding notation in Table~\ref{tab:notation-all} settings below.

\begin{table}[h]
    \centering\footnotesize{
    \begin{tabular}{l l l l}
\toprule
N. & Problem specifications & Notation & Section\\ 
\midrule
1 & general utilities, objective $f=u_{\min}$ or $f=\sum u_j$ & $1|u_j| f$ & \ref{sec:general},\ref{sec:algo-greedy}\\
2 & linear utilities, objective  $f\in\{u_{\min},\sum u_j\}$ & $1| u_j= b_j - a_jC_j | f$ & \ref{sec:general},\ref{sec:algo-greedy}\\
3 & release dates, linear utilities, objective  $f\in\{u_{\min},\sum u_j\}$  & $1| r_j, u_j= b_j - a_jC_j | f$ & \ref{sec:release}\\
4 & release dates, unitary proc. times,  lin. utilities, obj.  $f\in\{u_{\min},\sum u_j\}$ & $1| p_j=1, r_j,  u_j= b_j - a_jC_j | f$ & \ref{sec:identical_proc_times}\\
5 & release dates, equal proc. times,  lin. utilities, obj.    $f\in\{u_{\min},\sum u_j\}$ & $1| p_j=p, r_j,  u_j= b_j - a_jC_j | f$ & \ref{sec:identical_proc_times} \\
6 & due dates, lin. utilities, bounded \# of late jobs, obj.   $f\in\{u_{\min},\sum u_j\}$ & $1|\sum U_j \le k, u_j= b_j - a_jC_j | f$ &  \ref{sec:duedates}\\ 
& & &\\
7 & increase of  intercepts  & $1 | u_j = \b_j - a_j C_j | u_{\min}$ & \ref{sec:bound-modb}\\
8 & decrease of  slopes & $1 | u_j = b_j - \a_j C_j | u_{\min}$ & \ref{sec:bound-moda}\\
9 & variation of utilities with constant area & $1 | u_j = \b_j - \a_j C_j, \Ae_j =A_j | u_{\min}$ & \ref{sec:fixed-area}\\
10 & variation of utilities of a single agent $\jstar$, 
objective $u_{\jstar}$  & $1 | \tilde{u}_{\jstar} | u_{\jstar}$ & \ref{sec:modi-single}\\
 & & &\\ 
11 & resched. with new agent, bounded individual disruptions $\delta_j$ 
& $1|\mbox{\it resch}, \delta_j\ge 0|u_{n+1}$  & \ref{sec:single_agent_utility}\\
12 & resched.  with new agent, bounded aggregate disruptions $\delta_j$ 
& $1|\mbox{\it resch}, \textstyle\sum\delta_j\ge 0|u_{n+1}$ & \ref{sec:single_agent_utility}\\
13 & resched.  with new agent, schedule optimality preservation
& $1|\mbox{\it resch}, u_j \ge u_F |u_{n+1}$ & \ref{sec:single_agent_utility}\\ 
 & & &\\ 
14 & target sequence $\sigma_F$ (leader-follower setting), minimize increase of $b_j$ 
& $1 | \sigma_F = \sigma_T, u_j = \b_j - a_j C_j | \sum \b_j$  & \ref{sec:target}\\
15 & target sequence $\sigma_F$ (leader-follower setting), minimize decrease of $a_j$ 
& $1 | \sigma_F = \sigma_T, u_j = b_j - \a_j C_j | \sum -\a_j$ & \ref{sec:target}\\
16 & target sequence $\sigma_{\mbox{\tiny\it SYS}}$ (leader-follower setting), minimize decrease of $a_j$ & $1 | \sigma_{\mbox{\tiny\it SYS}}=\sigma_T, u_j = b_j - \a_j C_j |  \sum -\a_j$ & \ref{sec:target}\\
\bottomrule
\end{tabular} 
    }
    \caption{Description of addressed problem variants and adopted Graham notation.}
    \label{tab:notation-all}
\end{table}

\subsection{Problem with fixed parameters}\label{subsec:fixed}
The first set of problems we address concerns classical offline single-machine scheduling problems, in which agents are equivalent to standard jobs characterized by fixed, given parameters. In particular, we consider several variants of the problem, with or without release dates and due dates, with arbitrary, unit, or constant processing times, and under two different objective functions corresponding to fair and system-optimal solutions. The notation for these variants is summarized in the first six rows of Table~\ref{tab:notation-all}, while the associated algorithmic and complexity results are presented in Section~\ref{sec:fixed-param}.

\subsection{Variable linear utility functions}\label{sec:variable linear}

A second class of problems involves agents whose utility functions are linear, i.e.,
$ u_j = b_j - a_j C_j$, with $a_j \ge 0$,
where $C_j$ denotes the completion time of agent $j$. 
Certain parameters of these utility functions can be modified to improve either the max-min fairness or the total system utility. To ensure realistic adjustments, modifications are subject to constraints.  
We indicate modifiable parameters with a tilde (see lines 7--10 in Table~\ref{tab:notation-all}).  
In the following we formalize the types of variations we consider.

\smallskip\par\noindent
\textbf{Increasing the intercepts $b_j$}.
We introduce a model in which each intercept  $b_j$ is allowed to be increased by a non-negative amount $\db_j$ yielding the new intercept $\b_j=b_j+\db_j$, under a total budget constraint
\[
\sum_j \db_j \le B.
\]
 The modified utility functions take the form $u_j = (b_j + \db_j) - a_j C_j$. 
Formally, the optimization problem is
\[
\max_{\db_j \ge 0} \ \min_j \big( (b_j + \db_j) - a_j C_j \big) \quad \text{s.t. } \sum_j \db_j \le B.
\]
This problem is studied in  Section~\ref{sec:bound-modb} and is denoted as $1 | u_j = {\b}_j - a_j C_j | u_{\min}$. 

\smallskip\par\noindent
\textbf{Decreasing the slopes $a_j$}. 
Analogously,  we consider decreasing the slopes $a_j$ by non-negative amounts $\da_j$, subject to
\[
0 \le \da_j \le a_j, \quad \sum_j \da_j \le A,
\]
so that modified utilities with new slopes $\a_j=a_j-\da_j$ remain non-increasing in $C_j$.  
The modified utility functions become
$u_j = b_j - (a_j - \da_j) C_j$.
We address this problem in Section~\ref{sec:bound-moda} and denote it  as
$1 | u_j = b_j - {\a}_j C_j | u_{\min}$. 

\medskip\par\noindent \textbf{Variation with area constraint}.
\label{sec:intro-var-area} 
We consider a variation of the linear utility model in which both the intercepts $b_j$ and the slopes $a_j$ are allowed to vary simultaneously, but subject to an aggregate constraint on the expected utility value taken by each agent.

Assuming that a job completion time is a random variable with uniform distribution $C_j\sim U(0,P)$, then, by linearity, the expected utility value for agent $j$ is $\mathbb{E}[u_j] = \b_j - \a_j\mathbb{E}[C_j] = \b_j - \a_j \frac P 2$.
We are interested in choosing $\a_j$ and $\b_j$ while keeping $\mathbb{E}[u_j]$ constant, in particular equal to the original value $b_j - a_j \frac P 2$.
As the  area of the trapezoid bounded by the utility line and the horizontal axis over $[0,P]$ is
\[
A_j \;=\; \int_0^P \big( b_j - a_j t \big)\, dt
        \;=\; P\big(b_j - a_j\frac{P}{2}\big) \;=: \; P\big(\b_j - \a_j\frac{P}{2}\big) = \Ae_j,
\]
our constraint is equivalent to impose that the area of the trapezoid remain constant.

Additionally, to ensure that the linear utility functions remain non-increasing and non-negative over the entire interval of feasible completion times, we impose $\a_j \ge 0$ and 
\[
\b_j - \a_j t \;\ge 0, \quad \forall\, t\in [0,P],
\]
which is equivalent to $\b_j \ge \a_j P$.
In conclusion we denote the problem in scheduling notation as
$
1| u_j = \b_j - \a_j C_j,\ \Ae = A_j| u_{\min}$
and address it in  Section~\ref{sec:fixed-area}.

\medskip\par\noindent \textbf{Variations for a single agent}.\label{sec:pre-single-var}
We also consider the case where only a single agent, say agent $\jstar$, can adjust its utility function, while all other agents data remain fixed.  
The three scenarios are:
\begin{enumerate}[noitemsep,nolistsep]
    \item the intercept $b_{\jstar}$ can be increased;
    \item the slope $a_{\jstar} > 0$ can be decreased (while remaining non-negative);
    \item the area $A_{\jstar}$ can be increased.
\end{enumerate}
For each variation, we analyze how the utility $u_{\jstar}$ changes and how the optimal schedule adapts, either to improve $u_{\jstar}$ or the overall objective. We denote these three versions of the problem as $1 | u_{\jstar}  = \b_{\jstar} - a_{\jstar} C_{\jstar}  | u_{\bar{\jmath}}$, 
$1 | u_{\jstar}  = b_{\jstar} - \a_{\jstar} C_{\jstar} | u_{\bar{\jmath}}$, and 
$1 | u_{\jstar}  = \b_{\jstar} - \a_{\jstar} C_{\jstar}, \tilde{A} = A_{\jstar}   | u_{\bar{\jmath}}$ and study them in  Section~\ref{sec:modi-single}.

\subsection{Rescheduling from the perspective of a new agent} \label{sec:singleagentdefinition}

Analogously to the single-agent variations presented in Section~\ref{sec:pre-single-var}, we consider a third class of problems from the perspective of a \emph{new agent}. 

We assume that an initial problem has been solved optimally, yielding a schedule $\sigma_F$ with utility values $u_j(\sigma_F)$ for each job $j=1,\dots,n$ and a maximum-minimum utility $u_F := \min_j \{ u_j(\sigma_F) \}$.
We then consider a subsequent \emph{rescheduling stage}, in which a new agent, denoted $n+1$, is introduced. The goal is to determine a revised schedule $\sigma$ that 
(1)  maximizes the utility of the new agent $u_{n+1}(\sigma)$, and
(2) preserves certain fairness conditions for the original agents $j = 1, \dots, n$.

Formally, given the initial schedule $\sigma_F$ and the original utilities $u_j(\sigma_F)$, the problem is
\[
\max_{\sigma} \ u_{n+1}(\sigma)\]
 subject to fairness preserving  constraints for the remaining agents.

To enforce fairness, we allow the possibility of providing a non-negative \emph{compensation} $\beta_j \ge 0$ to each of the original agents, e.g., in the form of monetary or resource adjustments.
We specifically consider three types of fairness-preserving constraints on the original agents’ utilities (including their compensation):
\begin{enumerate}
[noitemsep,nolistsep]
    \item {Individual non-decrease:} $u_j(\sigma) + \beta_j \ge u_j(\sigma_F)$ for all $j=1,\dots,n$;
    \item {Aggregate non-decrease:} $\sum_{j=1}^{n} \big(u_j(\sigma) + \beta_j\big) \ge \sum_{j=1}^{n} u_j(\sigma_F)$;
    \item {Schedule optimality preservation:} the minimum utility among all $n$ original agents in the new schedule, including compensation, is at least $u_F$.
\end{enumerate}
These formulations are general: the same reasoning applies if the rescheduling perspective is taken with respect to one of the original $n$ agents, rather than a newly introduced agent.
We denote by $\delta_j  = u_j(\sigma) + \compj - u_j(\sigma_F)$ the utility disruption of any job $j$ in the new schedule $\sigma$ compared to the original solution $\sigma_F$, including compensations.

The three addressed version of the problem are then denoted by $1 | resch, \delta_j \geq 0 | u_{n+1}$, 
$1 | resch, \sum \delta_j \geq 0 | u_{n+1}$, and 
$1 | resch,u_j \geq u_F | u_{n+1}$.
We refer to Section~\ref{sec:single_agent_utility}
for computational and analytical results for this class of problems.

\subsection{Leader-follower utility adjustment}\label{subsec:stackel}

Finally,  we consider a bi-level optimization setting inspired by a leader-follower scenario. 
An entity, the leader, aims to enforce a certain schedule $\sigma_T$, for example, one that minimizes the sum of completion times to optimize throughput or maximizes total utility. 
In this setting, the follower still solves the resulting problem $1 | \ldots | u_{\min}$ to achieve a fair solution for all jobs, but the leader influences the outcome by modifying the utility functions of the jobs. 
In order to do so, the leader is allowed to control either the intercepts $\b_j = b_j + \db_j$, or the slopes $\a_j = a_j - \da_j$, analogous to the modifications described earlier. 

These changes $\db_j \geq 0$ and $0 \leq \da_j \leq a_j$ affect the sequence determined by the follower when computing the fair solution. 
We are therefore interested in the following problem: Minimize the total modification cost, i.e., either $\sum_j \db_j$ or $\sum_j \da_j$, subject to the condition that the fair solution $\sigma_F$ produced by the follower (who wants to maximize the minimum utility under the modified utility functions) corresponds to the desired target sequence $\sigma_T$.

We denote these two versions of the  problem as 
$1 | \sigma^* = \sigma_T, u_j = \b_j - a_j C_j | \sum \b_j$  and 
$1 | \sigma^* = \sigma_T, u_j = b_j - \a_j C_j | \sum -\a_j$ 
where $\sigma^*$ can be the fair solution $\sigma_F$ or the system optimum $\sigma_{\sys}$,
and study them in Section~\ref{sec:target}. (As it is clarified there, note that problem $1 | \sigma_{\sys} = \sigma_T, u_j = \b_j - a_j C_j | \sum \b_j$ is not meaningful.)

\subsection{Summary of contributions}

Across the sixteen (plus six, considering the two different objectives of the first six problems) problem variants considered in this work, we conduct a detailed analysis of their computational complexity and provide, whenever possible, efficient solution methods. In particular, for each variant we either establish NP-hardness or design polynomial-time and, in some cases, pseudopolynomial-time algorithms. 

Table~\ref{tab:contribution-all} offers an overview of our main contributions, highlighting the complexity classification and the type of algorithmic results obtained for each problem setting. This summary aims to provide a guide  through the diversity of models analyzed and the methodological advances achieved in the present study.

Note that 
all our algorithms in Sections~\ref{sec:fixed-param} and~\ref{sec:single_agent_utility} apply for general non-increasing utility functions (assuming that a function value can be computed in constant time).

\begin{table}[h]
    \centering\footnotesize{
\begin{tabular}{l l l}
\toprule
Problem & Complexity  for $f = u_{\min}$ & Complexity  for $f = \sum u_{j}$ \\
\midrule
$1|u_j| f$ & Polynomial (Thm.~\ref{th:greedy}) & Strongly NP-hard (Obs.~\ref{obs:maxsum}) \\
$1 | u_j = b_j - a_j C_j | f$ & Polynomial (Thm.~\ref{th:greedy}) & Polynomial \citep{smith1956}\\
$1 | r_j, \, u_j = b_j - a_j C_j | f$ & Strongly NP-hard (Thm.~\ref{th:release_strongNPH}) & Strongly NP-hard (Thm.~\ref{th:release_strongNPHsum}) \\
$1 | p_j = 1, \, r_j, \,  u_j | f$ & Polynomial (Thm.~\ref{th:pj=1}) & Polynomial: assignment (Obs.~\ref{th:pj=1sum}) \\
$1 | p_j = p, \, r_j, \,  u_j | f$ & Polynomial (Thm.~\ref{th:pj=p}) & Polynomial 
\cite[see][ Obs.~\ref{obs:BaptisteEqualLength}]{Baptiste2000} \\
$1 | \sum U_j \le k, \, u_j = b_j - a_j C_j | f$ & Polynomial  (Thm.~\ref{th:bounded-latejobs}) & Strongly NP-hard (Thm.~\ref{th:bounded-latejobssum}) \\
\bottomrule
\end{tabular}
    }
    \caption{
    Overview of results for problems 1--6 of Table \ref{tab:notation-all} 
    }
    \label{tab:contribution-all}
\end{table}

\section{Problems with fixed parameters}
\label{sec:fixed-param}
\begin{table}[h]
\centering\footnotesize{
\begin{tabular}{l l}
\toprule
Problem & Complexity   \\
\midrule
\addlinespace
$1 | u_j = \b_j - a_j C_j | u_{\min}$ & Polynomial (Sect.~\ref{sec:bound-modb}) 
\\
$1 | u_j = b_j - \a_j C_j | u_{\min}$ & Polynomial (Sect.~\ref{sec:bound-moda}) 
\\
$1 | u_j = \b_j - \a_j C_j, \, \Ae_j = A_j | u_{\min}$ & Polynomial (Sect.~\ref{sec:fixed-area}) 
\\
$1 | \tilde{u}_{i} | u_{i}$ & Polynomial (Sect.~\ref{sec:modi-single}) \\
\addlinespace
$1 | \mbox{\it resch}, \, \delta_j \ge 0 | u_{n+1}$ & Weakly NP-hard (Thm.~\ref{th:NA1-hard}) \\
$1 | \mbox{\it resch}, \, \textstyle\sum \delta_j \ge 0 | u_{n+1}$ & Weakly NP-hard (Thm.~\ref{th:NA2-hard}) \\
$1 | \mbox{\it resch}, \, u_j \ge u_F | u_{n+1}$ & Polynomial (Thm.~\ref{th:NA3-poly}) \\
\addlinespace
$1 | \sigma_F = \sigma_T, u_j = \b_j - a_j C_j | \sum \b_j$ & Polynomial (Sect.~\ref{sec:target}) \\
$1 | \sigma_F = \sigma_T, u_j = b_j - \a_j C_j | \sum -\a_j$ & Polynomial (Sect.~\ref{sec:target}) \\
$1 | \sigma_{\mbox\tiny{\it SYS}} = \sigma_T, u_j = b_j - \a_j C_j |\sum -\a_j$ &
Polynomial (Sect.~\ref{sec:target}) \\
\bottomrule
\end{tabular}
    }
    \caption{
    Overview of results for problems 7--16 of Table \ref{tab:notation-all} }
    \label{tab:contribution-after-split}
\end{table}



In this section, we consider the general case of the problem with given and fixed parameters, namely the first six problems of Table \ref{tab:notation-all}. 
In particular, we first propose a general solution approach that can be applied to problem $1|u_j|u_{\min}
$, for any continuous utility functions $u_j$ and then refine it by introducing a (faster) greedy algorithm suitable for several versions of the problem. 
Finally, we examine several variants, including the case with release dates, the case with identical processing times, and the case with due dates. 
For each of these variants, we provide corresponding complexity results and solution algorithms.

\subsection{General solution approach}\label{sec:general}

We can describe a generic solution approach based on binary search given lower and upper bounds \underbar{$u$} and $\bar{u}$ on the optimal solution value $u_F=u_{\min}(\sigma_F)$ (assuming continuous utility functions).
A general upper bound is given by
$\bar{u}:= \max_j \{u_j(0)\}$, while a lower bound may depend on additional problem data, such as due dates and release dates.
For linear utility functions we can choose the upper bound $\bar{u}:= \max_j \{b_j\}$.
We can perform a binary search in $\log (\bar{u}-\underbar{$u$})$ iterations, in each of which we have a target value $u^T$.

To check whether $u_F \geq u^T$, the completion time $C_j$ of each job $j$ has to fulfill $u_j(C_j) \leq u^T$, i.e., recalling Equation~\eqref{eq:definverse}, $C_j \leq u^{-1}_j(u^T)$, which implies a due date 
\begin{equation}\label{eq:target-duedate}
    d^T_j = u^{-1}_j(u^T)
\end{equation}
for each job.
If the underlying scheduling problem with due dates can be solved by some known algorithm in polynomial time this immediately translates into an optimal polynomial algorithm for our problem.
In fact, it is enough to establish the existence of a feasible solution for the problem with due dates without any tardy job.

The running time of such an algorithm is polynomial, but not strongly polynomial, as it contains the $\log$-factor of $\bar u$, which  depends on the length of the encoded input. Thus, for several problem variants we will present new and simpler algorithms with strongly polynomial running times, i.e., depending only on the number of input values.

\subsection{Greedy-type algorithm}\label{sec:algo-greedy}

Through the binary search approach of Section \ref{sec:general}, provided that the due dates in Equation \eqref{eq:target-duedate} can be computed in polynomial time, an optimal solution can likewise be found in polynomial time. In fact, the maximum lateness can be evaluated by sorting the jobs in non-decreasing order of their due dates (EDD rule), which enables an efficient check of whether $u_F \geq u^T$.
However, we can avoid binary search by the following strongly polynomial greedy-type algorithm which we will call \gre.

For the completion time of all unscheduled jobs $T$, initially set to $\sum_j p_j$, we pick the job $k$ with highest utility if completed at time $T$. 
The procedure iterates by considering at each step time $T:=T - p_k$.
It works for arbitrary non-increasing utility functions $u_j$.
A sketch of \gre\ applied for an example with three jobs and linear utilities is shown in Figure~\ref{fig:greedy}.
Note that this approach is quite similar to \emph{Lawler's Algorithm} well known for the minimization of the maximum of regular functions 
under precedence constraints~\citep{lawler73}.
Note that, as for Lawler's algorithm, our approach can be easily adapted to the case of precedence relationships among jobs.
\begin{algorithm}[htbp]
\begin{algorithmic}[1]
\STATE {\bf Input}: a set of agent-jobs $J =\{J_1,J_2,\ldots,J_n\}$
\STATE {\bf Output}: permutation (schedule) $\pi = \pi_1, \ldots,\pi_n$ ($\pi_i$ is the $i$-th job in the sequence
\STATE $T := \sum_j p_j$
\FOR{$i = n$ downto $1$}
\STATE pick the remaining job $k\in J$ with highest utility at time $T$ and schedule it at  position $i$ (i.e. $\pi_i := k = \arg\max\{u_j(T) : j\in J\}$)
\STATE $T:= T - p_k$ 
\STATE $J:=J\setminus\{k\}$
\ENDFOR
\STATE {\bf Return:} $\pi$
\end{algorithmic}
     \caption{{\sc Max-Min Greedy}}
\label{alg:greedy}
\end{algorithm}

	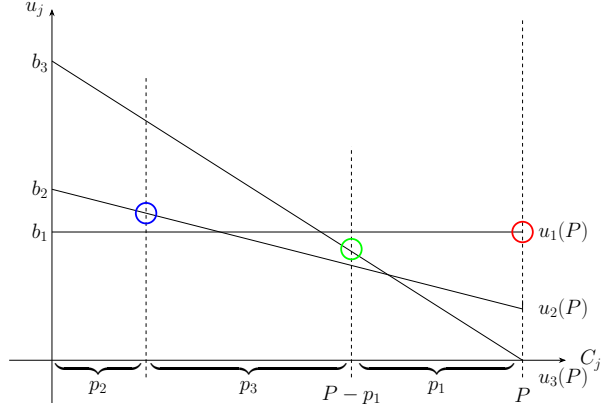
\begin{figure}[!ht]
		\centering
		\resizebox{0.5\textwidth}{!}{%
			\begin{circuitikz}
				\tikzstyle{every node}=[font=\LARGE]
				\node [font=\LARGE] at (15.0,6) {$P$};
				\draw [->, >=Stealth] (1.25,5.75) -- (1.25,17.25);
				\draw [->, >=Stealth] (0,7) -- (16.25,7);
				\draw [short] (1.25,15.75) -- (15,7);
				\draw [dashed] (15,6.5) -- (15,17);
				\draw [short] (1.25,10.75) -- (15,10.75);
				\draw [dashed] (10,6.5) -- (10,13.25);
				\draw [dashed] (4,6.5) -- (4,15.25);
				\draw [short] (15,8.5) -- (1.25,12);
				\draw [short] (15,8.75) -- (15,8.5);
				\node [font=\LARGE] at (0.75,17.25) {$u_j$};
				\node [font=\LARGE] at (17,7) {$C_j$};
				\node [font=\LARGE] at (16.2,10.75) {$u_1(P)$};
				\node [font=\LARGE] at (0.9,10.75) {$b_1$};
				\node [font=\LARGE] at (16.2,8.5) {$u_2(P)$};
				\node [font=\LARGE] at (0.9,12.0) {$b_2$};
				\node [font=\LARGE] at (16.2,6.5) {$u_3(P)$};
				\node [font=\LARGE] at (0.9,15.7) {$b_3$};
				\node [font=\LARGE] at (12.5,6.75) {$\underbrace{\mbox{\hspace{4.5cm}}}$};	
				\node [font=\LARGE] at (12.5,6.2) {$p_1$};
                \node [font=\LARGE] at (10,6) {$P-p_1$};
				\node [font=\LARGE] at (7,6.75) {$\underbrace{\mbox{\hspace{5.7cm}}}$};
				\node [font=\LARGE] at (7,6.2) {$p_3$};
				\node [font=\LARGE] at (2.6,6.75) {$\underbrace{\mbox{\hspace{2.5cm}}}$};
				\node [font=\LARGE] at (2.6,6.2) {$p_2$};
				\draw [color=red, line width=1.5pt ]  (15,10.75) circle (0.3cm);
				\draw  [color=green, line width=1.5pt]  (10,10.25) circle (0.3cm);			
				\draw  [color=blue, line width=1.5pt]  (4,11.3) circle (0.3cm);
			\end{circuitikz}
		}\caption{Algorithm~\gre\ applied to an instance of $1|u_j|u_{\min}$ with three jobs and linear utility functions.\label{fig:greedy}}
	\end{figure}

In conclusion we have the following result whose simple proof is reported in the Appendix.
\begin{theorem}\label{th:greedy}
If, for all $j\in J$,
$u_j$ is any non-increasing function of the completion time $C_j$ of job $j$, then 
\gre\ computes an optimal solution of \ $1|u_j|u_{\min}$ in $O(n^2)$ time.
\end{theorem}
%




\subsection{Release dates}\label{sec:release}

In this section we consider our problem when release dates are present, namely $1 | r_j,u_j | f$. 
We explore the case in which \emph{each} job has an arbitrary given release date and the special case in which all jobs \emph{but one} have null release dates.

\medskip\par\noindent \textbf{Fair solutions}.\label{sec:release-fair}
Let us first focus on the fair objective, i.e., $f = u_{\min}$. 
In what follows, we show that introducing release dates makes the problem strongly NP-hard, even when all utilities are linear. The result is obtained by reducing a utility target value to an equivalent due-date formulation. A detail proof is reported in the Appendix.

\begin{theorem}\label{th:release_strongNPH}
Problem $1|r_j,u_j|u_{\min}$ 
is strongly NP-hard, even for linear utility functions.
\end{theorem}
%
%
It can furthermore be shown that the problem remains NP-hard even when only a single job has a release date strictly greater than zero. In this restricted setting, the problem is in fact weakly NP-hard. A proof of this statement is given in the Appendix.

\begin{theorem}\label{th:release_binary}
Problem $1|r_{j},u_j|u_{\min}$ 
is weakly NP-hard, even if the  utility functions are linear and only 
one release date is nonzero.
\end{theorem}

To show that this special case is indeed only weakly NP-hard, we briefly describe a pseudopolynomial dynamic programming algorithm.

\begin{algorithm}[hp]
\caption{Pseudo-polynomial dynamic programming algorithm for 
the special case of $1|r_j,u_j|u_{\min}$
with only one nonzero release date}
\label{alg:dp-special}
{\small
\begin{algorithmic}[1]
\STATE \textbf{Input}: Jobs $1,\ldots,n+1$, processing times $p_j$, release date $r$ for job $n{+}1$
\STATE \textbf{Output}: Maximum minimum utility $u_{\min}^*$
\STATE Let $P:= \sum_{j=1}^n p_j$, $u_F:=0$
\FOR{$s_{n+1} = r$ \TO $r + P$}
    \STATE Perform binary search on $u^T$
 \WHILE{there is a candidate value $u^T$}
         \STATE $\mathcal{S} \leftarrow \emptyset$
        \STATE Compute due dates $d_j^T  = u_j^{-1}(u^T)$ for all jobs $j$
        \STATE Identify set $A$ of jobs with $d_j^T \le s_{n+1} + p_{n+1} + p_j - 1$, set $P_A:=\sum_{j\in A}p_j$ 
        \IF{$P_A\leq s_{n+1}$}
        \STATE $\mathcal{S} \leftarrow \{(P_A,0)\}$
        \STATE Sort remaining jobs in non-decreasing order of $d_j^T$
        \FORALL{jobs $j$ in sorted order}
            \STATE $\mathcal{S}' \leftarrow \emptyset$
            \FORALL{states $(P_1, P_2) \in \mathcal{S}$}
                \IF{$P_1 + p_j \le s_{n+1}$}
                    \STATE Add $(P_1 + p_j, P_2)$ to $\mathcal{S}'$
                \ENDIF
                \IF{$d_j^T \geq s_{n+1} + p_{n+1} + P_2 + p_j$}
                    \STATE Add $(P_1, P_2 + p_j)$ to $\mathcal{S}'$
                \ENDIF
            \ENDFOR
            \STATE $\mathcal{S} \leftarrow \mathcal{S}'$
        \ENDFOR
        \ENDIF
        \IF{$\mathcal{S} = \emptyset$}
            \STATE Choose smaller candidate $u^T$
        \ELSE
            \STATE Choose larger candidate $u^T$
        \ENDIF
    \ENDWHILE
    \STATE $u_F:= \max\{u_F, u^T\}$
\ENDFOR
 \STATE \textbf{Return} $u_F$
\end{algorithmic}
}
\end{algorithm}

Let job $n+1$ be the job with a release date $r$, while all other jobs are available from the beginning. Let $P=\sum_{j=1}^n p_j$ be the total job processing time not including job $n+1$.
We perform an iteration over all feasible starting times $s_{n+1}$ of job $n+1$, i.e., $s_{n+1}=r, r+1, \ldots, r+P$.
For each starting time $s_{n+1}$ we perform binary search similar to the one described in Section~\ref{sec:general}.
For each candidate value $u^T$ we obtain, recalling Equation \eqref{eq:definverse}, due dates $d_j^T = u_j^{-1}(u^T)$ for every job $j$ and check if there exists a feasible schedule where all jobs are on time w.r.t.\ these due dates, i.e., where the maximum lateness is $\leq 0$.
To do so, we have to decide for each job whether it should be scheduled before job $n+1$ finishing before $s_{n+1}$ or after $n+1$ starting not before $s_{n+1}+p_{n+1}$.
Within each of the two resulting subsets, jobs must be scheduled in EDD (earliest due date) order to minimize the maximum lateness.
 All jobs with $d_j^T \leq s_{n+1}+p_{n+1}+p_j-1$ must be scheduled before $n+1$, and thus can be omitted from further consideration after recording their total processing time $P_A$.
For all remaining jobs we perform dynamic programming where each state $(P_1, P_2)$ consists of jobs with total processing time $P_1$, resp.\ $P_2$, scheduled before, resp.\ after job $n+1$.
After initializing a single state $(P_A,0)$, we go through the jobs in EDD order.
Each job $j$ is added to each existing state $(P_1, P_2)$ either in the first or the second subset.
This replaces $(P_1, P_2)$ by two new states $(P_1+p_j, P_2)$ and $(P_1, P_2+p_j)$, where the former is 
generated if $P_1+p_j \leq s_{n+1}$ and the latter if $d_j^T \geq s_{n+1} + p_{n+1} + P_2 + p_j$.
Clearly, the total number of states is bounded by $P^2$.
If we reach a job which cannot be added to any existing state, 
we have established that there is no solution with starting time $s_{n+1}$ for job $n+1$ reaching the desired utility target and have to lower $u^T$.
Otherwise, every state generated by adding the final job is a feasible solution with minimum utility at least $u^T$ and we can try to increase $u^T$. 
The procedure sketched above is further clarified by the pseudo-code provided in Algorithm~\ref{alg:dp-special}.

The overall pseudopolynomial running time is $O(P \cdot n P^2 \cdot \log(UB-LB))$ where $UB$ and $LB$ are, respectively, upper and lower bounds on the objective values.

\medskip\par\noindent \textbf{System optimal solution}.
In the following we consider problem $1|r_{j},u_j|\sum_j u_j$ with release dates and global objective function, i.e., we consider the maximization of the total utility.
It turns out that this problem has the same complexity as the one with maximum fairness objective.
Since completion times can be mapped to utilities, a simple reduction from the strongly NP-hard
problem $1|r_j | \sum Cj$ proves the following statement.
\begin{theorem}\label{th:release_strongNPHsum}
Problem $1|r_{j},u_j|\sum_j u_j$ 
is strongly NP-hard, even if the  utility functions are linear.
\end{theorem}

\begin{theorem}\label{th:release_weakNPHsum}
Problem $1|r_{j},u_j|\sum_j u_j$ 
is weakly NP-hard, even if the  utility functions are linear and only 
one release date is nonzero.
\end{theorem}
\begin{proof}
\def\p{\bar p}
We use a reduction from \textsc{Partition}. 
We consider an instance $I_P$ with integers
$w_1,\ldots,w_n$ and $B = \frac12 \sum_{j=1}^n w_j$. Correspondingly,
we construct an instance $I$ of $1|r_{j}, u_j = b_j - a_jC_j| \sum_j u_j$
with $n{+}1$ jobs. 
For each $j=1,\ldots,n$, let $p_j = a_j = w_j$ and $r_j=0$,
and let job $n{+}1$ have processing time $p_{n+1} = \p$, $a_{n+1} = M$
(for $M$ sufficiently large) and release date $r_{n+1} = B$.
Since $\sum_{j=1}^n p_j = 2B$, any feasible schedule has makespan at least $2B+\p$.
A feasible schedule for $I$ such that $n+1$ is processed in the interval $[B,B+\p]$ and having a makespan equal to $2B+\p$ will be called \emph{balanced} (see schedule $\tilde\sigma$ in Figure~\ref{fig:proofThm4}).
We show that $I_P$ is a YES instance if and only if an optimal schedule of $I$
is balanced. 

Hereafter, for any $S\subseteq \{1,\ldots, n\}$, we use $p(S)= \sum_{j\in S}p_j =\sum_{j\in S}w_j$.

\begin{figure}[htb]
    \centering
    \includegraphics[width=0.7\linewidth]{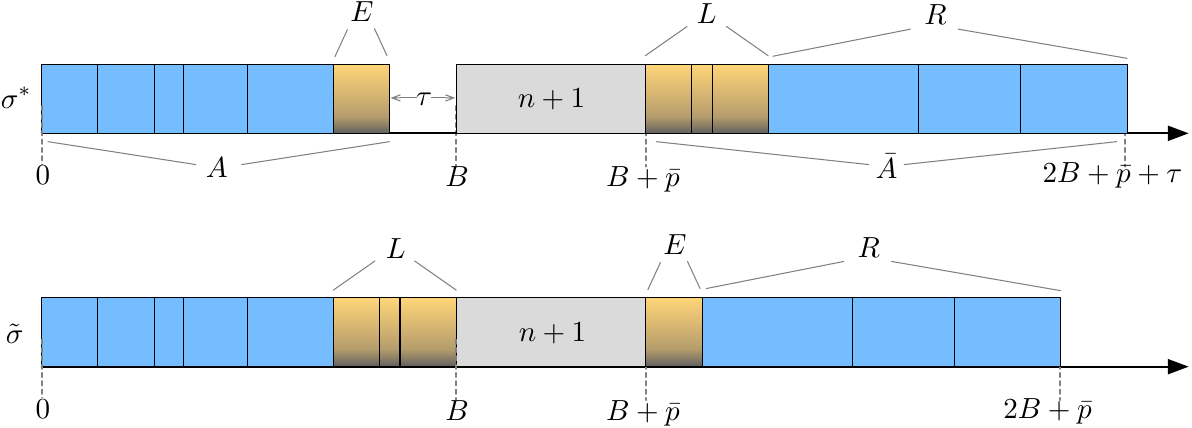}
    \caption{Proof of Theorem~\ref{th:release_weakNPHsum}: Non-balanced schedule $\sigma^*$ and balanced schedule $\tilde\sigma$ obtained by swapping blocks $E$ and $L$. }
    \label{fig:proofThm4}
\end{figure}

Assume that $I$ admits an optimal schedule $\sigma$ that is balanced. By definition of balanced schedule, there is a set $S\subseteq \{1,\ldots, n\}$ of jobs executed in $[0,B]$, and $p(S) = B$. Hence $I_P$ is a YES instance.

Conversely, assume that $I_P$ is a YES instance. Then there exist a subset $S\subseteq \{1,\ldots,n\}$ with $p(S) = B$ and a feasible schedule $\sigma$ where the jobs in $S$ are processed in the interval $[0,B]$, job $n{+}1$ in $[B,B+\p]$, and the set $\bar S$ of remaining jobs in $[B+\p,2B+\p]$. Such a schedule is therefore balanced. 
Note that since $a_j = p_j$, any internal permutation of jobs within $S$ or within $\bar S$ does not alter the overall utility. Moreover, by standard interchange arguments, it is easy to prove that 
all balanced schedules have identical total utility. 

We now show that, if a balanced schedule exists then it is optimal for $I$. 
Again, under the assumption that $I_P$ is a YES instance and hence a balanced schedule $\sigma$ exists for $I$, consider an optimal schedule $\sigma^*$ for $I$ which is not balanced. 
Let $A$ and $\bar A$ be the sets of jobs process before and, respectively, after job $n+1$ in  $\sigma^*$. 
As $a_{n+1}=M$ is very large, $\sigma^*$ would schedule $n+1$ to start as early as possible, i.e., at its release time $r_{n+1}=B$, and there is idle time of length $\tau \ge 1$ in $[0,B]$.

Since a balanced solution exists, there must be two subsets of jobs $E\subset A$  (possibly empty) and $\emptyset\neq L\subset \bar A$, such that $p(L) = p(E) + \tau$. With no loss of generality (possibly by re-arranging the jobs within $A$ and $\bar A$ with no loss of utility) we may assume that $E$ and $L$ are scheduled, respectively, as a block of the last consecutive jobs in $A$ and a block of the first consecutive jobs of $\bar A$. Let $R$ be the set of remaining jobs processed after $L$ (i.e., $R=\bar A\setminus L$). Now, consider a schedule $\tilde{\sigma}$ obtained by swapping the blocks $E$ and $L$ and shifting all jobs in $R$ earlier by $\tau$. Then $\tilde{\sigma}$ has makespan $2B+\p$ and is therefore balanced. Figure~\ref{fig:proofThm4} illustrates the two schedules $\sigma^*$ and $\tilde\sigma$.

The change in total utility when transforming $\sigma^*$ into $\tilde{\sigma}$ equals
\[
\Delta u
= \underbrace{-\, p(E)\bigl(\p+\tau+p(E)\bigr)}_{\mbox{\scriptsize loss for postponing $E$}}
  + \underbrace{p(L)\bigl(\p+p(L)\bigr)}_{\mbox{\scriptsize gain for advancing $L$}}
  + \underbrace{p(R)\tau.}_{\mbox{\scriptsize gain for advancing $R$}}
\]
Since $p(L) = p(E) + \tau$, simple algebra yields
\[
\Delta u
= \bigl(p(E)+\p+\tau+p(R)\bigr)\tau
> 0.
\]
Hence $u(\tilde{\sigma}) > u(\sigma^*)$, contradicting the optimality of $\sigma^*$.
In conclusion, if $I_P$ is a YES instance, an optimal schedule for $I$ is balanced and this concludes the proof.
\end{proof}
Similar to the fair-objective case discussed in the previous section, when we
focus on the special case in which only one job has a non-zero release date,
the problem $1 \mid r_j, u_j \mid \sum_j u_j$ is weakly NP-hard (in contrast to
the general case, which is strongly NP-hard; see
Theorem~\ref{th:release_weakNPHsum}).  
This follows from the fact that one can design a pseudopolynomial dynamic
programming algorithm analogous to the one sketched above for 
$1 \mid r_j, u_j \mid u_{\min}$.

The algorithm iterates over all feasible starting times $S$ of  job $n+1$, which is the only job with a positive release date $r$.  For each candidate starting time $s_{n+1}$, we must decide for every job whether to schedule it before $s_{n+1}$ or after $s_{n+1} + p_{n+1}$.  
Within each of the two resulting subsets, jobs must be scheduled in non-increasing order of $\frac{a_j}{p_j}$, in accordance with Expression~\eqref{eq:sum-utili}.
A dynamic programming algorithm considers the jobs in this order and proceeds
analogously to the case $1 \mid r_j, u_j \mid u_{\min}$.  
Starting from a single initial state $(0,0)$, the algorithm processes the jobs
one by one.  For each job $j$, any existing state $(P_1, P_2)$ is replaced by
the two successor states $(P_1 + p_j,\, P_2)$ and/or $(P_1,\, P_2 + p_j)$.  
After all jobs have been processed, the state yielding the maximum total
utility is retained as the best solution associated with the candidate value
$s_{n+1}$.  
Optimality follows from the prescribed sorting of the jobs.

%

\medskip\par\noindent \textbf{Unit or identical processing times}.\label{sec:identical_proc_times}
A well-studied special case in the scheduling literature \citep{Baptiste99,Baptiste2000,bib:CHEN2026_equallenghtmultiagent,lawler78,simons1978} considers instances in which all processing times are identical, that is, $p_j = p$ for all jobs $j$, or in the even more restricted case, $p_j = 1$.
It is important to note that these two assumptions are not equivalent.
In what follows, we show that both special cases $1|p_j=1,r_{j},u_j| u_{\min}$ and $1|p_j=p,r_{j},u_j | u_{\min}$ admit polynomial-time algorithms, although the structural properties of their corresponding optimal solutions differ.

When all processing times are unitary, any solution of $1|p_j=1,r_j,u_j|u_{\min}$  can be represented solely by the assignment of each job to a  unitary time slot corresponding to the completion/starting time of the job.  Clearly, this time slot must be feasible, i.e. each job $j$ can be scheduled only in the time slots between its (integer) release date $r_j$ and the time slot corresponding to the schedule makespan. 

In searching for an optimal solution of  $1|p_j=1,r_j,u_j|u_{\min}$, it is also clear that we may restrict our attention to semi-active schedules in which the makespan is minimum and unnecessary idle times are avoided. 
Note that the minimum makespan value $T$ can be easily computed---and so are the available time slots for the jobs in any semi-active schedule---as shown below (see Algorithm~\ref{alg:completion_new1}). 

This property allows for a straightforward algorithmic approach: The solution of the scheduling problem corresponds to a bipartite matching between the set of jobs and the time slots identified by the minimum makespan schedule, where each job can only be assigned to time slots not earlier than its release date.
The profit of assigning job $i$ to time slot $C_j$ is given by $u_i(C_j)$.
Hence, solving $1|p_j=1,r_j,u_j|u_{\min}$ reduces to a linear bottleneck assignment problem, which can be solved in $O(n^{2.5}/\sqrt{\log n})$ time \cite[see Sec.~6.2,][]{assign2012}.

\begin{algorithm}[htbp]
{\small
\begin{algorithmic}[1]
\STATE {\bf Input:} set of jobs $J = {J_1,J_2,\ldots,J_n}$ with release dates $r_1, r_2, \ldots, r_n$
\STATE Sort and rename jobs in non-decreasing order of release dates, i.e., $r_j \le r_{j+1}$ for $j = 1, \ldots, n-1$ \COMMENT{This ordering remains valid for all successive subsets}
\STATE $pred(j+1):=j$ \COMMENT{Defines the current job sequence}
\STATE $C := 0$
\FOR {$j = 1$ to $n$}\label{l:for}
\STATE $C := C_j := \max\{r_j + 1, C + 1\}$
\ENDFOR\label{l:endfor}
\STATE $T := C_n$
\STATE $J' := J$ \COMMENT{Set of unscheduled jobs}
\STATE Let $last$ denote the index of the last job in $J'$\label{l:last_1}
\REPEAT
\STATE $j'' := j' := last$ \COMMENT{$j'$ is the candidate job to be scheduled last}
\WHILE{$C_{j''} - 1 > r_{j''}$}
\STATE $j'' := pred(j'')$
\IF{$u_{j''}(T) > u_{j'}(T)$}
\STATE $j' := j''$
\ENDIF
\ENDWHILE\label{l:last_2}
\STATE Schedule $j'$ at position $|J'|$ with $C_{j'} := T$\label{l:sched_j}
\STATE $J' := J' \setminus \{j'\}$\label{l:update_J}
\STATE Move all jobs $\ell \in J'$, $\ell > j'$, earlier by 1, and update $C_\ell$, %
$pred$, 
and $last$
\STATE $T := C_{last}$ \COMMENT{Marks the endpoint of a new block if applicable}\label{l:update_T}
\UNTIL{$|J'| = 1$}
\STATE Schedule the remaining job in $J'$ at position $1$
\end{algorithmic}
}
\caption{Exact algorithm for $1|p_j=1,r_j,u_j| u_{\min}$}
\label{alg:completion_new1}
\end{algorithm}

We can develop a more efficient procedure based on the same principle as Algorithm \gre, introduced in Section~\ref{sec:algo-greedy}.
This procedure is detailed in Algorithm~\ref{alg:completion_new1}.
First, the jobs are sorted in non-decreasing order of their release dates, and a semi-active schedule with minimum makespan $T$ is constructed by starting each job as soon as it is released and all its predecessors in the sorted order have been completed (see lines~\ref{l:for}--\ref{l:endfor} of Algorithm~\ref{alg:completion_new1}).

Then, starting from the last available time slot and proceeding iteratively (lines~\ref{l:last_1}--\ref{l:last_2}), the algorithm identifies a job $j'$ with the highest utility at that time slot among those that can be moved to the end of the sequence without increasing the makespan.
This relocation is performed only if it is feasible, that is, only if all jobs succeeding $j'$ in the current sequence can be moved earlier.
Such a shift is possible only when none of these jobs starts exactly at its release date.
The while-loop terminates, because the first job is always started at its release date.
After identifying such a job $j'$, it is moved to the end of the sequence (line~\ref{l:sched_j}), and the set of remaining jobs $J'$ and their associated parameters are updated accordingly (lines~\ref{l:update_J}--\ref{l:update_T}).
If idle times are present in the original sequence, each resulting block of jobs is implicitly handled; note that such a block always begins with a job starting exactly at its release date.
It is straightforward to verify that the running time of Algorithm~\ref{alg:completion_new1} is $O(n^2)$.
Hence, we obtain the following result.

\begin{theorem}\label{th:pj=1}
The problem with release dates and unit processing times $1|p_j=1,r_j,u_j|u_{\min}$ can be solved in $O(n^2)$ time.
\end{theorem}

As in the previous case, it is easy to show that determining a system–optimal solution under unit processing times and general utility functions can be straightforwardly reformulated as a linear sum assignment problem. Hence, the following statement holds:
\begin{observation}\label{th:pj=1sum}
    The problem with release dates and unit processing times $1|p_j=1,r_j,u_j|\sum u_j$ can be solved in polynomial time.
\end{observation}

By scaling, the same reasoning also applies to the case $p_j = p > 1$ 
\emph{if $p$ is a common factor of all release dates}.
However, this does not hold for the case of constant but arbitrary processing times $p_j = p > 1$, which is discussed below.

\medskip
The special case of $1|p_j=1,r_j,u_j|u_{\min}$ where $u_j(C_j) = b_j - C_j$ for given constants $b_j \in \mathbb{Z}{+}$, $j = 1, \ldots, n$, is equivalent to the classical scheduling problem $1|p_j=1,r_j|L_{\max}$. For integer-valued release dates, this problem can be efficiently solved in $O(n \log n)$ time using the so-called preemptive EDD rule \cite[see Sect.~3.2 in][]{elementsched}. 
In the same paper, the authors further observed that the situation where all processing times $p_j$ are equal to some arbitrary positive integer $p$ is fundamentally different from the case where $p_j = 1$ for all jobs. Although one can rescale the time to make $p = 1$, this transformation would generally turn the release dates into rational numbers. This violates the underlying assumption under which it is claimed that the preemptive EDD rule solves the problem $1|p_j=1,r_j|L_{\max}$.

A similar situation applies for our problem. Indeed, if all processing times are equal, but not unitary, it may make sense to postpone already released jobs and thereby accrue ``unforced'' idle times. Thus, the optimal schedule may be a not semi-active schedule.
This aspect is illustrated in Figure~\ref{fig:example unitary} by an example with two jobs and the following data:
$r_1=0$, $r_2=p/2$, 
$u_1=3p-C_1$,  $u_2=2p-C_2$.
The minimum makespan is $C_{\max}=2p$, reached by assigning first job $1$ followed by job $2$ starting at time $p$.
For this schedule there is $u_1=2p$ and $u_2=0$.
Scheduling job $2$ first at time $r_2=p/2$ and job $1$ later, namely at time $3p/2$, leaves the machine idle from time $0$ to $p/2$ although job $1$ is already released.
It gives $u_1=u_2=p/2$ which is a better objective function value.

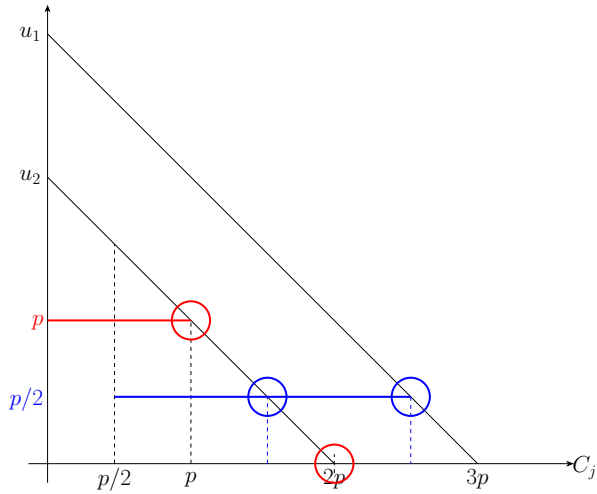
\begin{figure}[!ht]
	\centering
	\resizebox{0.5\textwidth}{!}{%
		\begin{circuitikz}
			\tikzstyle{every node}=[font=\huge]
			\node [font=\LARGE] at (0.8,15.75) {$u_1$};
			\node [font=\LARGE] at (0.8,12) {$u_2$};
			\draw [->, >=Stealth] (1.25,4) -- (1.25,16.5);
			\draw [->, >=Stealth] (0.75,4.5) -- (15.0,4.5);
			\draw [short] (1.25,15.75) -- (12.5,4.5);
			\draw [short] (1.25,12) -- (8.75,4.5);
			\node [font=\LARGE] at (15.3,4.4) {$C_j$};
			\draw [dashed] (5,4.5) -- (5,8.25);
			\draw [dashed] (8.75,4.25) -- (8.75,4.75);
			\node [font=\LARGE] at (5,4.1) {$p$};
			\node [font=\LARGE] at (8.75,4.1) {$2p$};
			\node [font=\LARGE] at (12.5,4.1) {$3p$};
			\node [font=\LARGE] at (3,4.1) {$p/2$};
			\draw [dashed] (3,4.5) -- (3,10.25);
			
			\node [color=blue,font=\LARGE] at (0.7,6.2) {$p/2$};
			
			\draw [color=blue,dashed] (7,4.5) -- (7,6.25);
			\draw [color=blue,line width=1.5pt] (3,6.25) -- (10.75,6.25);
			\draw [color=blue,dashed] (10.75,6.25) -- (10.75,4.5);
			\draw [color=blue,line width=1.5pt] (7,6.25) circle (0.5cm);
			\draw [color=blue,line width=1.5pt] (10.75,6.25) circle (0.5cm);
            			
				\node [color=red,font=\LARGE] at (1,8.2) {$p$};
					
				\draw [color=red,line width=1.5pt] (1.25,8.25) -- (5,8.25);
				\draw [color=red,line width=1.5pt] (5,8.25) circle (0.5cm);
				\draw [color=red,line width=1.5pt] (8.75,4.5) circle (0.5cm);
                
		\end{circuitikz}
	}\caption{Example to illustrate that idle times may be useful.
    Schedule $\langle 1,2 \rangle$ without idle times in red, optimal schedule $\langle 2,1 \rangle$ with idle time in blue.\label{fig:example unitary} }
\end{figure}

Therefore, the iterative greedy approach described above is no longer applicable, as it relies on knowing the total completion time of all jobs.
Nevertheless, a polynomial-time solution can still be obtained by performing a binary search over the objective function values and exploiting a result from the literature to solve the resulting subproblem. 
In particular, Carlier developed an algorithm for the problem of minimizing the total number of late jobs with release dates and equal processing times $1|p_j=p,r_j|\sum U_j$, which runs in $O(n^3\log n)$ time~\citep{Baptiste99,Carlier81}.
Accordingly, for a given target value $u^T$ within the feasible range of objective values, we can repeatedly apply Carlier's algorithm with job due dates defined as in Expression~\eqref{eq:target-duedate}, and combine this with the general binary search framework presented in Section~\ref{sec:general} to address our problem variant. In conclusion, the following result holds.

\begin{theorem}\label{th:pj=p}
    Problem $1|p_j=p,r_j,u_j| u_{\min}$ 
    can be solved in polynomial time $O(n^3\log n\log(UB-LB))$ where $UB$ and $LB$ are, respectively, upper and lower bounds on the objective values.
\end{theorem}

Note that the problem under consideration can be reformulated as the standard scheduling problem of minimizing the maximum value of a regular objective function in the presence of release dates and identical processing times. 
In particular, the above result generalizes the work of 
\citet[][Sect.~3.3]
{elementsched}, who studied the specific case of minimizing maximum lateness, i.e., $1|p_j = p,r_j|L_{\max}$. They proposed an algorithm with running time $O(n^3 \log^2 n)$, which is tailored to the $L_{\max}$ objective only. 
Their approach relies on a more sophisticated technique, namely an iterative application of the earliest due date (EDD) rule, which allows the problem to be decomposed and solved efficiently under that objective.

\medskip
As for the corresponding system-optimal solution $\sigma_{\sys}$, \citet{Baptiste2000} presents a $O(n^5)$ dynamic programming algorithm to minimize total weighted completion time for single machine scheduling problems with release dates and equal-length processing times. We may then state:
\begin{observation}\label{obs:BaptisteEqualLength}
    Problem $1|p_j=p,r_j,u_j =b_j-a_jC_j| \sum u_{j}$ can be solved in polynomial time. 
\end{observation}

\subsection{Due dates}\label{sec:duedates}

We now turn our attention to the problem in which each job $j$ is associated also with a due date $d_j$.  
Hereafter, we use the well-known notation $U_j$ to denote whether job~$j$ is late or not, where $U_j = 1$ if job $j$ is late (i.e., $C_j > d_j$) and $U_j = 0$ otherwise.

We first assume that there exists a schedule in which all jobs can meet their due dates, that is, the EDD rule produces a schedule without late jobs.
In this case, the entire schedule contains no idle time.  
For this setting, a simple variation of algorithm~\gre\ can be applied. 
We consider all due dates as deadlines and, at each iteration corresponding to a time step $T$, we restrict attention to the unscheduled jobs with deadlines satisfying $d_j \ge T$.  
The correctness of this approach follows from the same arguments used in the proof of Theorem~\ref{th:greedy}, applied to the restricted set of feasible jobs $J$ at each time step.

\begin{theorem}\label{th:no-latejobs}
Problem  $1 | \sum U_j \le k, \, u_j = b_j - a_j C_j | u_{\min}$ , in the special case where all jobs can be completed on time (i.e., $k = 0$), is polynomially solvable.  
Moreover, an optimal schedule can be computed in $O(n^2)$ time.
\end{theorem}

\medskip
We now consider the case in which every feasible schedule necessarily contains late jobs.  
In this situation, we may assume that all late jobs receive a constant utility (strictly smaller than the utility obtainable before the due date, and possibly equal to~$0$).  
This setting can be represented by a non-increasing (not necessarily continuous) utility function, and thus the result of Theorem~\ref{th:greedy} continues to apply.

\medskip
Next, we consider the more standard due-date setting in which all late jobs are
discarded, that is, they contribute no utility to the objective function.
To prevent an excessive ``waste of jobs'', we impose an upper bound $k$
on the number of late jobs (equivalently, a lower bound of $n - k$ on the
number of jobs that must be completed on time).  
Clearly, this requirement is meaningful only if $k \ge U^*$, where $U^*$
denotes the minimum possible number of late jobs as can be obtained using Moore's
algorithm \citep{moore1968latejobs}.
We can now apply a variation of the general approach described in Section~\ref{sec:general}.  
For a given target utility value $u^T$, we determine, for each job $j$, the latest completion time at which job $j$ still finishes on time while achieving utility at least $u^T$.
Accordingly, we define updated due dates 
$d^T_j :=\min \{d_j, u_j^{-1}(u^T)\}$ 
meaning that any job completed by time $d^T_j$ is both on time and satisfies the target utility requirement.

We then apply Moore's algorithm (running in $O(n \log n)$ time) to the instance
with modified due dates $d^T_j$, yielding $k^T$ late jobs.
If $k^T > k$, we decrease the target utility $u^T$, whereas if $k^T \le k$,
we may increase it.   
Hence, we obtain the following result for the problem variant with a bounded
number of late jobs.

\begin{theorem}\label{th:bounded-latejobs}
    Problem $1 | \sum U_j \le k, \, u_j = b_j - a_j C_j | u_{\min}$  i.e., determining the minimum utility attained by any on-time job under the
    constraint that at most $k$ jobs may be tardy, can be solved in
    $O(n\log n\log(UB-LB))$ time, where $UB$ and $LB$ are, respectively, upper and lower bounds on the objective value.
\end{theorem}

Following the previous theorem, it is natural to ask whether an analogous polynomial-time result can be obtained when the objective is to maximize the  total utility. 
Somewhat surprisingly, the situation is now fundamentally different from what happens with other versions of the problem. In fact,  while the
min-utility problem with at most $k$ tardy jobs is solvable in polynomial time,
the corresponding problem of seeking a system optimum turns out to be
strongly NP-hard.

This problem is equivalent to minimizing the total weighted completion time subject to deadline constraints.  
\citet{ChenPottsWoeginger1998} state that minimizing the weighted sum of completion times with deadlines is strongly NP-hard, and they attribute this to \citet{bib:lenstra1977}.  
However, the paper by Lenstra et al.\ establishes only weak NP-hardness  via a reduction from the knapsack problem. 
For the setting considered here, we provide a direct proof showing that the problem is indeed strongly NP-hard.
 
\def\M{M}
\def\k{\kappa}
\def\MM{L}

\begin{theorem}\label{th:bounded-latejobssum}
Problem $1 | \sum U_j \le k, \, u_j = b_j - a_j C_j | \sum u_j$ that is, maximizing the total utility of on-time jobs under the constraint
    that at most $k$ jobs may be tardy, is strongly NP-hard.
\end{theorem}
\begin{proof}
In order to prove the thesis, we reduce {\sc 3-Partition} to our problem. So, given an instance $I_{3P}$ of {\sc 3-Partition}, with 3$n$ integers, $w_1,\ldots,w_{3n}$, such that $\sum_{j=1}^{3n} w_j = W = n T$ and $T/4 < w_j < T/2$ for $j=1,\ldots,3n$, we generate an instance $I$ of problem {$1|\sum U_j \le k, u_j = b_j - a_j C_j | \sum u_j$} as follows.

In instance $I$ there are $4n-1 + {k}$ jobs and the maximum number of allowed late jobs is $k\ge 0$.
The idea, in the reduction, is that the first $3n$ jobs (``item-jobs'') correspond to the integers of $I_{3P}$, the next $n-1$ jobs are used for the reduction as ``separators'' of the schedule, and the last $k$ ``long'' jobs are those who are late in any feasible schedule.
So, in instance $I$ we have
\[ p_{j} = \left\{ \begin{array}{ll}
    w_j & \mbox{for } j=1, \ldots, 3n  \\
    \M & \mbox{for } j=3n + 1, \ldots, 4n-1 \\
    \MM & \mbox{for } j=4n, \ldots, 4n-1+k
\end{array}
\right .\]
where $\M > 0$ and $\MM \ge nT + (n-1)\M =: E$ are given integers. 
Moreover, 
\[ d_{j} = \left\{ \begin{array}{ll}
    E & \mbox{for } j=1, \ldots, 3n  \\
    (j-3n)(T+\M) & \mbox{for } j=3n + 1, \ldots, 4n-1 \\
    \MM & \mbox{for } j=4n, \ldots, 4n-1+k.
\end{array}
\right .\]
Note that at most one long job can meet its deadline; however, if this happens it would be the only on-time job in the schedule, since $L \ge E$. 
This would make the solution infeasible, since more than $k$ jobs would then be late. 
As a consequence, in any feasible solution the $k$  long jobs will be the late ones, while all the other $4n -1$ jobs (item and separation jobs) must be on time. 

The jobs' utility functions for instance $I$ are defined as
\[ u_{j}(C_j) = \left\{ \begin{array}{ll}
   H - w_j C_j & \mbox{for } j=1,2, \ldots, 3n  \\
0 & \mbox{for } j=3n + 1, \ldots, 4n-1+k
\end{array}
\right .\]
with $H$ large enough to guarantee nonnegative utilities of the item-jobs over all the 
schedule's span (e.g., $H \ge \max_j\{w_j\}\cdot E$).

Observe that any feasible schedule $\sigma$ for $I$ can be seen as a sequence of $n$ (possibly empty) blocks $B_i \subseteq \{1,\ldots,3n\}$, $i=1,\ldots, n$, of item-jobs divided by the $n-1$ separators $3n+1,\ldots,4n-1$, followed by the $k$ late long jobs:
$
\langle B_1, 3n+1, B_2, 3n+2, \ldots, 4n-1, B_n, 4n, \ldots, 4n - 1 +k\rangle.
$
Recalling that the contribution of the late jobs and that of the separators to the total utility is null, the objective function value for schedule $\sigma$ can be expressed as $\sum_{j=1}^{3n} u_j(\sigma) = 3nH - \sum_{j=1}^{3n} w_j C_j$. Hence, an optimal solution for $I$ is obtained minimizing $\sum_{j=1}^{3n} w_j C_j$. 

Let $W_i = \sum_{j \in B_i} w_j$ denote the total processing time of block $i$, and let $s_i$ denote its starting time. 
We write $i \prec j$ (respectively, $i \preccurlyeq j$) to indicate that job $i$ precedes (respectively, precedes or is equal to) job $j$ in the schedule $\sigma$. Then:

{\small
\begin{align}\label{eq:obj-proof}
    &\sum_{j=1}^{3n} w_j C_j = \sum_{i=1}^n \sum_{j\in B_i} w_j \big(s_i + \sum_{\ncoppa{\ell\in B_i}{\ell\preccurlyeq j}} w_{\ell}\big) =\sum_{i=1}^n \bigg\{\sum_{j\in B_i} w_j \bigg((i-1)M +  \sum_{\ell\in \cup_{r=1}^{i-1} B_{r}} w_{\ell}\bigg)  + w_j \sum_{\ncoppa{\ell\in B_i}{\ell\preccurlyeq j}} w_{\ell}\bigg\} =
    \\
    &= \sum_{i=1}^n \bigg\{\bigg(M(i - 1)\sum_{j\in B_i} w_j \bigg) +  \sum_{j\in B_i} \bigg( w_j\!\! \sum_{\ncoppa{\ell\in \{1,\ldots,3n\}}{\ell\preccurlyeq j}} w_{\ell}\bigg)\bigg\} = \\
    &= M\ \sum_{i}^n (i-1)W_i + \sum_{j=1}^{3n} w_j^2 + \sum_{j=1}^{3n}\sum_{\ncoppa{\ell=1}{\ell < j}}^{3n} w_j\, w_{\ell} = M\ \sum_{i=1}^n (i-1)W_i + \frac 1 2 \bigg(\sum_{j=1}^{3n} w_j^2 + W^2\bigg).\label{eq:obj-last}
\end{align}
}
Note that the double sum $\sum_j\sum_{\ell} w_j w_{\ell}$ in~\eqref{eq:obj-last} is equal to the sum of the ${3n}\choose 2$ products $w_i w_j$ with $i$ and $j$ distinct jobs.
Moreover, observe that the last term $\tfrac{1}{2}(\cdots)$ in the same equation 
is independent of the composition of the blocks and depends only on the specific values of $w_j$. 
Therefore, for a given instance $I$ with fixed $w_j$ values, the decision problem reduces to determining the composition of the blocks $B_i$ so as to minimize the term $\sum_{i=1}^n (i-1) W_i$. 

Due to the increasing coefficients in the summation $\sum_{i=1}^n (i-1) W_i$, an optimal solution packs as much units of processing time as possible in the earlier blocks, while ensuring that the separators are not late.
\begin{figure}[ht]
    \centering
    \includegraphics[width=0.8\linewidth]{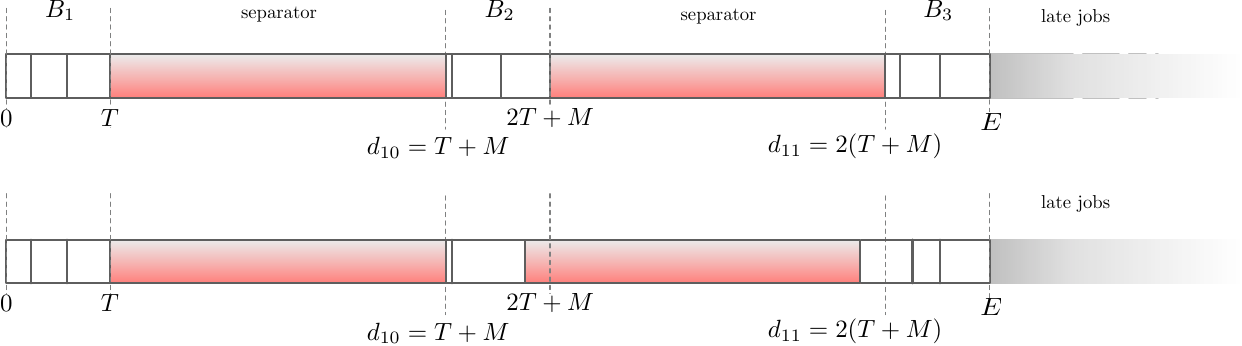}
    \caption{Optimal schedules corresponding to a {\sc yes}-instance (above) and
    a {\sc no}-instance of {\sc 3-Partition} with $n=3$.}
    \label{fig:thm11fig}
\end{figure}

If $I_{3P}$ is a {\sc yes}-instance, then we are in an ideal configuration (illustrated by the upper chart in Figure~\ref{fig:thm11fig}) in which we may set $W_i = T = \frac W n$ for all $i=1,\ldots,n$.
Thus, each block is assigned exactly the length necessary to force every
separator job in the sequence to complete exactly at its designated deadline.
In this case, by \eqref{eq:obj-last}, the value $u_{\mathrm{YES}}$ of an optimal
solution $\sigma_{\sys}$ to instance $I$, under the assumption that $I_{3P}$ is
a \textsc{Yes}-instance, is given by:
\begin{eqnarray}
    u_{\mbox{\scriptsize YES}} &=& 3nH - \tfrac{1}{2}\!\left(\sum_{j=1}^{3n} w_j^2 + W^2\right) - M \frac W n \sum_{i=1}^n (i-1) \nonumber\\
    &=&3nH - \tfrac{1}{2}\!\left(\sum_{j=1}^{3n} w_j^2 + W^2\right) - \frac 1  2 (MW(n-1)).
    \label{eq:obj-yes-n}
\end{eqnarray}
If $I_{3P}$ is a \textsc{no}-instance, then there must exist a block $i \in \{1,\ldots, n\}$ such that the sequence of blocks consists of $W_{\ell} = T$, for $\ell = 1,\ldots,i-1$, and $W_i < T$ (see the lower chart in Figure~\ref{fig:thm11fig} for $i=2$).
This means that at least one unit of processing time must be shifted from block $i$ to a later block, which is associated with a strictly larger coefficient in $\sum_{j=1}^{3n} w_j C_j$.
Because of the deadlines assigned to the separator jobs, any increase in $\sum_{j=1}^{3n} w_j C_j$ cannot be compensated in later blocks: the amount by which a block can exceed $T$ is limited  by the slack created when earlier blocks
finish before $T$.
Thus, the optimal utility $u_{\sys}$ decreases by at least $M$ compared to $u_{\mbox{\scriptsize YES}}$ in~\eqref{eq:obj-yes-n}.

In conclusion, $I_{3P}$ is a \textsc{yes}-instance if and only if the optimal utility value $u_{\sys}$ of the corresponding instance $I$ of
$1 | \sum U_j \le k | \sum u_j$
is not larger than $u_{\mbox{\scriptsize YES}}$.
\end{proof}

\section{Adjustable linear utility functions}
\label{sec:variable}

The idea of modifying or adjusting problem parameters in order to improve
system performance has been extensively studied in the scheduling and
optimization literature under the general framework of \emph{controllable
processing times} or, more broadly, \emph{controllable parameters}.
In these models, one can alter certain job characteristics, typically at some cost or within a given budget, in order to achieve a more desirable schedule or objective value \citep{Controllable_LearningEffectBiskup,Controllable_DueDateAssign,Controllable-group,Controllable-position,Controllable-TimeDeptSched,Controllable-Total,ControllablePT_Shabtay}.

Motivated by this line of research, we study how controlled adjustments to the utility functions of individual jobs can improve the overall performance criterion.  Throughout this section, we restrict attention to linear utility functions of the form $u_j = b_j - a_j C_j$, with $a_j \ge 0$ for every job $j$ which represent decreasing valuations as completion times increase.
We first examine the setting in which a fixed global bound is imposed on the total amount of permissible modification.  We then turn to a complementary framework in which utilities may be adjusted subject to a fixed ``area'' constraint, reflecting situations where improvements must be distributed in a balanced manner among the agents.

Across the next two sections, we consider two natural types of modifications:
increasing the intercepts $b_j$, and decreasing the negative slopes $a_j$.
Both operations enhance the utilities of on-time jobs, and our goal is to raise
the maximum minimum utility achieved across all agents.  Since such
modifications cannot be made arbitrarily large, we assume a global budget or
capacity constraint that limits the total amount of allowed change.
Note that the natural analogy of these problems with fairness replaced by efficiency can be solved in a rather straightforward way.

\subsection{Modifications of the intercepts $b_j$}
\label{sec:bound-modb}

{\bf Problem $1 | u_j = {\b}_j - a_j C_j | u_{\min}$}.
At first, we consider increasing each $b_j$ by some $\db_j$ such that the total increase is bounded by the given budget bound $B$, i.e., $\sum_j \db_j \leq B$ with $\db_j \geq 0$.

There is an iterative algorithm to apply this increase (with the technical assumption that $a_j>0$) as follows.
Starting with $\db_j=0$ for $j=1,\ldots,n$, we define $J$ as the set of all jobs  attaining the lowest utility value in the current schedule.
In the beginning, we will usually have $|J|=1$, say $J=\{j'\}$. (It is trivial to extend the procedure in the case $|J|>1$.)
We start by increasing $\db_{j'}$ until the increased utility of $j'$ matches the second largest utility value of some job $j''$, or the increase $\db_{j'}=B$, in which case we stop. 
Otherwise, insert $j''$ into $J$.
We proceed by increasing all $\db_{\ell}$ for $\ell \in J$ uniformly until the increased utility reaches the lowest utility of any job not in $J$, which is then added to $J$, or until $\sum_{\ell\in J} \db_\ell =B$.
The latter case would end the procedure, in the former case we iterate.

Note that this modification process does not change the sequence of jobs as given initially by \gre.
To see that consider any job $j'$ for which $\db_{j'}$ is increased.
For any job $i \not\in J$ with $C_{i} < C_{j'}$ the algorithm did not schedule $i$ at time $C_{j'}$ since $u_{i}(C_{j'}) \leq  u_{j'}(C_{j'})$.
Clearly, increasing $\db_{j'}$ will only strengthen this condition.
For any job $k \not\in J$ with $C_{k} > C_{j'}$ the increase of $\db_{j'}$ would stop (at the latest) when $u_{j'}(C_{j'}) =  u_{k}(C_{k})$ is reached.
Since utilities are decreasing functions also after this increase we still have $u_{k}(C_{k}) \geq  u_{j'}(C_{k})$ and the decision of \gre\ remains unchanged.\footnote{In the special case $a_{j'}=0$ it would be beneficial to change the ordering of $j'$ and $k$.}
Comparing job $j'$ with any job in $J$ it is clear that the utilities of both jobs are increased uniformly and thus their differences in utility remain unchanged for any completion time.

This approach works not only for linear utility functions but also for arbitrary utilities, where the increase $\db_{j}$ is added to the given utility $u_j(C)$.
The corresponding procedure will be utilized in the proof of Theorem~\ref{th:NA3-poly}.

While the above procedure can be easily performed in $O(n^2)$ time, we can also do better.
Given a sequence of jobs, the problem actually asks for a \emph{break} job $j_b$, such that for all jobs $j$ with $u_j \leq u_{j_b}$ the intercepts $b_j$ are increased to reach a common utility value at least $u_{j_b}$, while all jobs $j$ with $u_j > u_{j_b}$ remain unchanged.
This task can be solved in $O(n)$ by a variant of the linear time median algorithm (proceed in the same way as finding the split item for the knapsack problem in linear time, \cite[Sect.~3.1]{KePfPi04}

\medskip
It should be noted that if we allow also $\db_j < 0$ the problem will always end up with a solution where all utility values are equal. 
Only the level of this joint utility value depends on $B$.
Assuming otherwise, we could reduce the value $\db_j$ for any job $j$ with utility larger than the current minimal utility and use the resulting excess in the budget constraint to increase the $\db_\ell$ values for all $\ell \in J$, thus improving the minimum utility.
Also in this case the sequence of the jobs is not affected by reducing some $b_j$ value. 
This is trivial for jobs completing after job $j$ and also clear for jobs completing earlier than $j$ since $j$ is chosen as the job with highest utility.
If its decrease would reach a tie with jobs scheduled earlier, these jobs have a higher utility at their current position due to the nonincreasing utility functions and thus the algorithm would move stopped decreasing $j$ and proceeded with those jobs instead.

\subsection{Modifications of the slopes $a_j$}
\label{sec:bound-moda}

{\bf Problem $1 | u_j = {b}_j - {\a}_j C_j | u_{\min}$}.
We apply the same concept to a decrease of the slopes $a_j$ by considering values $\da_j \geq 0$ with $\sum_j \da_j \leq A$ and modified utilities given by $u_j = b_j - \a_j C_j = b_j - (a_j-\da_j) C_j$.
To preserve the property of nonincreasing utilities we will impose $\da_j \leq a_j$.

We can give an iterative algorithm quite similar to the above case where we chose the $\db_j$ values.
As before, we start with the job $j'$ having the lowest utility value in the schedule given by \gre\ and increase $\da_{j'}$ until the increased utility of $j'$ reaches some $j''$.
The same procedure as above applies although the increase for the $\da_{\ell}$ values for $\ell \in J$ is not uniform.
To reach any new common lowest utility value $\bar u$, one has to set $\da_\ell := a_\ell - \frac{b_\ell - \bar u}{C_\ell}$ for each $\ell \in J$.
Note that one can stop the procedure as soon as some job $\ell \in J$ reaches $\da_\ell = a_\ell$, since at this point the utility of job $\ell$ cannot be improved anymore.

The above arguments, that the sequence of jobs is not affected by these modifications, still hold.
For jobs $i$ this is trivial, for $k$ the above argument holds for $b_{j'} \geq u_{k}(C_{k})$, 
while for
$b_{j'} < u_{k}(C_{k})$,
it is clear that job $j'$ can never reach the utility of $k$ and thus $k$ will never be included in $J$.

\medskip
Similar to the case of intercepts, we can also consider the extension of the problem to the case where $\da_j<0$ is allowed.
As above, we would expect to reach a solution where all utility values are equal. 
However, this situation cannot be reached in all cases, since $b_j$ is an upper bound for $u_j(C)$ for every $j$ reached for $\da_j = a_j$.
%
The idea is to first increase all $a_j$ until all utilities are equal to $u_F$ and increase the budget $B$ by the accrued increases of the slopes. 
Then all slopes are simultaneously decreased and thereby the common utility value for all jobs are increased (consuming different $\da_j$ for every job $j$), until the budget runs out or one job reaches $\da_j = a_j$.

\subsection{Modifications of agent utilities with fixed area}
\label{sec:fixed-area}

If we allow both the intercepts $b_j$ and the slopes $a_j$ to vary, it makes sense to impose additional restrictions on the shape of the linear utility function. While we enforce that the modified utility $\tilde u_j$ remains nonnegative for all completion times up to $P$, we also impose an upper bound on the expected utility for agent $j$.

As mentioned in Section~\ref{sec:intro-var-area}, the (varied) area $\Ae_j = P(\b_j - \a_j (P / 2))$  (which forms a trapezoid for every choice of a linear function)  can be interpreted as a measure of expected utility return for agent $j$, under the assumption that the unknown variable $C_j$ is uniformly distributed over the interval $[0, P]$. Note also that the slope $\a_j$ can be viewed as a measure of the risk attitude of agent $j$. Smaller (resp.\ larger) values of $\a_j$ indicate an agent which is more risk-averse (resp.\ risk-tolerant).

In our setting, we enforce that $\Ae_j$ cannot exceed the original expected utility return $A_j$. 
Since every agent will fully exploit the given area $A_j$, i.e., $\Ae_j := A_j$  we have that 
$\b_j = \frac {A_j} P + \frac{\a_jP} 2$ while $\a_j\in[0,2\frac {A_j} {P^2}]$ to keep $u_j \ge 0$. 
It is easy to see that every line of a utility function must pass through the midpoint $(P/2, A_j/P)$.
Looking at each job separately, it is clear that for a job $j$ with completion time $C_j\geq P/2$ it would be the best choice to set $\a_j=0$ and $\b_j=A_j/P$.
On the other hand, for $C_j\leq P/2$ a triangular shape maximizes the utility, i.e.\ setting $\b_j=2A_j/P$ as large as possible and $\a_j=2A_j/P^2$, such that $\u_j(P)=0$.
Summarizing, the linear utility function $\u_j(C)$ chosen for a job $j$ in a schedule $\sigma$ can be expressed as
\begin{align}\label{eq:utilarea}   
\u_j(C_j(\sigma)) = \left\{\begin{array}{ll}
     \frac{2A_j}{P}(1 - \frac{C_j(\sigma)}{P})& \mbox{if } C_j(\sigma) \le \frac P 2;\\[\smallskipamount]
     \frac{A_j}{P} &  \mbox{if } C_j(\sigma) > \frac P 2\,.
\end{array}\right.
\end{align}

\begin{figure}[!hp]
	\centering
	\resizebox{0.4\textwidth}{!}{%
		\begin{circuitikz}
			\tikzstyle{every node}=[font=\LARGE]
			\draw [->, >=Stealth] (1.25,4) -- (1.25,17);
			\draw [->, >=Stealth] (0.75,4.5) -- (17.5,4.5);
			\draw [short] (1.25,14.5) -- (16.25,4.5);
			\draw [dashed] (16.25,4.25) -- (16.25,10.75);
			\draw [dashed] (8.75,4.25) -- (8.75,10.75);
			\draw [short] (8.75,9.5) -- (16.25,9.5);
			\draw [dashed] (1.25,9.5) -- (8.75,9.5);
			
			\node [font=\LARGE] at (16.25,3.75) {$P$};
			\node [font=\LARGE] at (8.75,3.75) {$P/2$};
			\node [font=\LARGE] at (0.35,9.5) {$A_j/P$};
			\node [font=\LARGE] at (0.25,14.5) {$2A_j/P$};
			\node [font=\LARGE] at (17.9,4.4) {$C_j$};
			\node [font=\LARGE] at (1.25,17.3) {$\u_j$};
			\draw [color=red, line width=1.5pt] (1.25,14.6) -- (8.7,9.65);
			\draw [color=red, line width=1.5pt] (8.7,9.6) -- (16.2,9.6);	
		\end{circuitikz}
	}\caption{Function $\u_j(C_j(\sigma))$ as defined in (\ref{eq:utilarea}).\label{fig:utilarea}}
\end{figure}
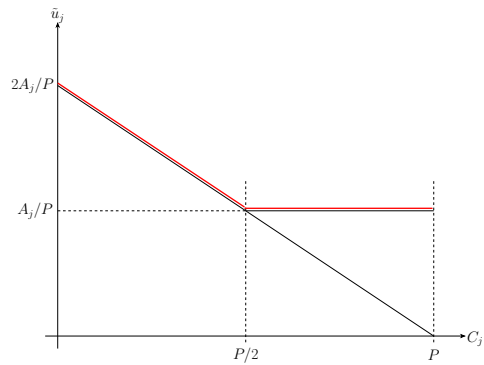
Therefore, we can represent each utility function by the maximum of the two cases, constant or triangular, as depicted in Figure~\ref{fig:utilarea}.
Thus, for problem $1 | u_j = \b_j - \a_j C_j, \Ae_j =A_j | u_{\min}$  the agents' utility functions are non-increasing with completion times and we immediately get the optimal solution by applying \gre.
In fact, this means that jobs are scheduled in the order of non-decreasing areas $A_j$.

\subsection{Modification for a single job}\label{sec:modi-single}

Given an instance of $1 | u_j = b_j - a_j C_j | u_{\min}$ and an optimal solution computed by \gre\ we consider the situation where a single job/agent, say $\jstar$, has the possibility of modifying its utility function, while all other jobs remain unchanged.
For every change of $u_{\jstar}$, the solution of the resulting problem is again computed by \gre.

\medskip
We will first discuss the case where the intercept $b_{\jstar}$ of job $\jstar$ can be increased by $\db_{\jstar} \geq 0$.
Clearly, a marginal increase $\db_{\jstar}$ will directly increase $u_{\jstar}$ by the same amount $\db_{\jstar}$.
However, at some point the increase of $\db_{\jstar}$ will imply a change in the decision of \gre\ for some job $k$ scheduled later than $\jstar$ in the original schedule. 
This will happen for the minimal value of $\db_{\jstar}$ such that there exists a job $k$ with $C_{k} > C_{\jstar}$ and $b_{\jstar} + \db_{\jstar} - a_{\jstar}C_{k} = u_{k}(C_{k})$.
Clearly, $k$ can be any of the jobs with $C_{k} > C_{\jstar}$ and not necessarily the successor of $\jstar$.

At this point, job $\jstar$ assumes the position of $k$ in the schedule which sets its completion time equal to $C_{k}$ and causes a sudden decrease of utility by $a_{\jstar}(C_{k}-C_{\jstar})$.
Also the jobs previously scheduled between $\jstar$ and $k$ will be reconsidered by \gre\ and possibly scheduled differently, but that will not affect $\jstar$.

Proceeding in this way, we can determine the \emph{modified utility function} $\tilde u_{\jstar}(\db_{\jstar})$ giving the utility of job $\jstar$ with intercept $b_{\jstar} + \db_{\jstar}$ as 
it results from \gre\footnote{Be aware that $\tilde u_{j}(\db_{j})$ expresses the utility as a function of the intercept change, while $u_j(C_j)$ gave the utility of a job depending on its completion time.}.
It is a sawtooth function with slope $1$ and discontinuous points whenever $\jstar$ reaches a new completion time, having a local maximum at this point, see Figure~\ref{fig:modify-intercept}.
At some point, $\jstar$ is relegated to the last position in the schedule with $C_{\jstar}=P$ and thus $\tilde u_{\jstar}(\db_{\jstar})$ increases infinitely with slope $1$.

\begin{figure}[!ht]
	\centering
	\resizebox{0.55\textwidth}{!}{%
		\begin{circuitikz}
			\tikzstyle{every node}=[font=\LARGE]
			\draw [->, >=Stealth] (1.25,2.5) -- (1.25,17.25);
			\draw [->, >=Stealth] (0,3.25) -- (25,3.25);
			\draw [short] (1.25,9.5) -- (3.75,12);
			\draw [short] (3.75,7) -- (7.5,10.75);
			\draw [short] (7.5,8.25) -- (12.5,13.25);
			\draw [short] (12.5,10.75) -- (16.25,14.5);
			\draw [short] (16.25,9.5) -- (18.75,12);
			\draw [short] (18.75,10.75) -- (23.75,15.75);
			\draw [dashed] (3.75,7) -- (3.75,12);
			\draw [dashed] (7.5,8.25) -- (7.5,10.75);
			\draw [dashed] (12.5,10.75) -- (12.5,13.25);
			\draw [dashed] (16.25,9.5) -- (16.25,14.5);
			\draw [dashed] (18.75,10.75) -- (18.75,12);
			\node [font=\LARGE] at (25.75,3.25) {\Huge $\db_{\jstar}$};
			\node [font=\LARGE] at (1.25,17.75) {\Huge $\tilde u_{\jstar}$};
		\end{circuitikz}
	}
\caption{Increased utility $\tilde u_{\jstar}$ as a function of the increase of the intercept $\db_{\jstar}$.}
\label{fig:modify-intercept}
\end{figure}

We are interested in an increasing sequence of local maxima. 
Clearly, an increase of $\db_{\jstar}$ only makes sense for job $\jstar$ if it leads to a new local maximum with higher value than all previously determined local maxima.
It may also be natural to consider an upper bound $B$ and ask for the best utility achievable for $\db_{\jstar} \leq B$.
This is either the last maximum of this increasing sequence obtained without exhausting $B$, or it arises for $\tilde u_{\jstar}(B)$.

A similar approach works for decreasing the negative slope $- a_{\jstar}$ by $\da_{\jstar}$.
In this case, the corresponding modified utility function remains a sawtooth function with discontinuous drops whenever the completion time of $\jstar$ changes.
Since $\tilde u_{\jstar}(\da_{\jstar})=b_{\jstar} - a_{\jstar}C_{\jstar} + \da_{\jstar}C_{\jstar}$,
the slope of this function is $C_{\jstar}$ which increases with every new position of $\jstar$ in the schedule as depicted in Figure~\ref{fig:modify-slope}.
Clearly, we impose the upper bound $\da_{\jstar} \leq a_{\jstar}$ which also yields the maximum of the modified utility function.

	\begin{figure}[!ht]
	\centering
	\resizebox{0.7\textwidth}{!}{%
		\begin{circuitikz}
			\tikzstyle{every node}=[font=\LARGE]
			\draw [->, >=Stealth] (-11.25,-1) -- (-11.25,15.25);
			\draw [->, >=Stealth] (-12,0.25) -- (22,0.25);			
			\node [font=\LARGE] at (22.75,0.25) {\Huge $\da_{\jstar}$};
			\node [font=\LARGE] at (-11.0,15.75) {\Huge $\tilde u_{\jstar}$};			
			
			\draw [short] (-11.25,2) -- (-5,3.25);
			\draw [short] (-5,1.25) -- (2.5,3.75);
			\draw [short] (2.5,0.75) -- (11.25,5.75);
			\draw [short] (11.25,3.25) -- (16.25,8.5);
			\draw [short] (16.25,5.75) -- (21.25,13.25);
			\draw [dashed] (-5,1.25) -- (-5,3.25);
			\draw [dashed] (2.5,0.75) -- (2.5,3.75);
			\draw [dashed] (11.25,3.25) -- (11.25,5.75);
			\draw [dashed] (16.25,5.75) -- (16.25,8.5);
		\end{circuitikz}
	}
\caption{Increased utility $\tilde u_{\jstar}$ as a function of the increase of the slope $\da_{\jstar}$.}
\label{fig:modify-slope}
\end{figure}

A third version of modifying the utility function considers the increase of the area of the utility function $A_{\jstar}$ by $\dA_{\jstar}$, assuming that all other agents $j$ remain at their areas $A_j$.
Recall that \gre\ results in a simple sorting of the jobs in increasing order of $A_j$.
The breakpoint of the utilities at $P/2$, as discussed in Section~\ref{sec:fixed-area}, implies that as soon as $C_{\jstar}\geq P/2$, the modified utility function $\tilde u_{\jstar}(\dA_{\jstar})$ increases linearly with slope $1/P$.
For $C_{\jstar}< P/2$ we observe again a sawtooth function, with positive slope $\frac{2}{P}(1 - \frac{C_{\jstar}}{P})$.
As $C_{\jstar}$ increases whenever $A_{\jstar}+\dA_{\jstar}$ matches one of the other areas $A_j$, the slopes of the linear pieces decrease, converging towards $1/P$ when $C_{\jstar}$ goes to $P/2$.
A visualization of this behavior is given in Figure~\ref{fig:modify-area}.

	\begin{figure}[!ht]
		\centering
		\resizebox{0.7\textwidth}{!}{%
			\begin{circuitikz}
				\tikzstyle{every node}=[font=\LARGE]
				\draw [->, >=Stealth] (-11.25,-1) -- (-11.25,15.25);
				\draw [->, >=Stealth] (-12,0.25) -- (24,0.25);			
				\node [font=\LARGE] at (24.75,0.25) {\Huge $\dA_{\jstar}$};
				\node [font=\LARGE] at (-11.0,15.75) {\Huge $\tilde u_{\jstar}$};			
\draw [short] (-11.25,3.25) -- (-7.5,8.25);
\draw [short] (-7.5,5.75) -- (-2.5,9.5);
\draw [short] (-2.5,5.75) -- (2.5,8.75);
\draw [short] (2.5,7) -- (10,10.75);
\draw [short] (10,8.25) -- (15,10.75);
\draw [short] (15,7) -- (24.75,12);
\draw [dashed] (-7.5,5.75) -- (-7.5,8.2);
\draw [dashed] (-2.5,5.75) -- (-2.5,9.5);
\draw [dashed] (2.5,7) -- (2.5,8.75);
\draw [dashed] (10,8.25) -- (10,10.75);
\draw [dashed] (15,7) -- (15,10.75);
\node [font=\LARGE] at (12.75,5.25) {\Huge slope $\frac 1 P$};
\end{circuitikz}}
\caption{Increased utility $\tilde u_{\jstar}$ as a function of the increase of the area $\dA_{\jstar}$.}
\label{fig:modify-area}
\end{figure}

\section{Rescheduling to maximize the utility of a new agent}\label{sec:single_agent_utility}

In this section we consider a two-stage optimization problem where in the first stage the standard problem  $1 | u_j | u_{\min}$ with $n$ agents is solved by \gre\ on a set of $n$ jobs.
This step yields an optimal schedule $\sigma_F$ with utility values $u_j(\sigma_F)$ for each job $j$ and the maximum minimum utility $u_F:=\min_j \{u_j(\sigma_F)\}$.

In the second stage we consider the arrival of a \emph{new agent} $n+1$ and ask for a new schedule $\sigma$ comprising the new agent's job as well as all the original jobs $1,\ldots, n$.
The new agent has a high priority and should be given a maximal utility (i.e., should be completed as early as possible), under a certain restriction for the utility of the original jobs.
Clearly, scheduling an additional job at a preferred position may worsen the utilities of other jobs.
Therefore, we introduce a budget $R$ which can be used to \emph{compensate} agents (jobs) by increasing (or, more generally, keeping some restrictions on) their utilities.
We will identify budget payment with utility increase and assume that each agent $j$ can be given an arbitrary nonnegative apportionment $\compj \ge 0$ to be added to its utility as long as $\sum_{j=1}^n \compj \leq R$.

This problem aligns with a few preceding studies about \emph{rescheduling}, in which 
a certain schedule must be rearranged in order to optimize a given objective function under bounded changes with respect to the original schedule.  
For papers on single machine rescheduling problems see, e.g.,
\citet{HallPotts2004}, \citet{NicosiaPacificiPferschy2021}, and 
\citet{RenerSalassaTkindt2025}. 
In these works, the authors define different measures of \emph{disruption} associated to a job, such as the difference in the positions (or in the completion times) between the original and the reordered schedules.


As the agents are pursuing the maximization of their utilities, here we measure {disruption}  of one agent $j$ as the difference among the utilities in the new schedule $\sigma$ (including possible compensations) and the original solution $\sigma_F$, i.e., 
\begin{align}\label{def:Delta}
\delta_j = u_j(\sigma) + \compj - u_j(\sigma_F).
\end{align}
We consider three reasonable versions of our rescheduling problem with different restrictions on agents' disruptions:
\begin{description}
    \item[Problem $1|\mbox{\it resch}, \delta_j\ge 0|u_{n+1}$]. 
    In this problem, we impose (in some sense) the strongest requirement since after including the new job, each single agent in $\sigma$ should have a utility (including compensation) at least as large as before in $\sigma_F$. 
    Recalling Equation~\eqref{def:Delta}, we can easily recast the constraints in terms of agents' disruption, as follows:
    $u_j(\sigma) + \compj \ge u_j(\sigma_F) \Leftrightarrow   \delta_j \geq 0$ and write the problem as
    \begin{align}
    &\max_{\sigma} \left\{ u_{n+1}(\sigma)\ :\  \delta_j \geq 0,\ j=1,\ldots,n\right\}.\label{NA1}
    \end{align}
    \item[Problem $1|\mbox{\it resch}, \textstyle\sum\delta_j\ge 0|u_{n+1}$]. Here, we relax the requirements of the previous problem \eqref{NA1} by bounding the sum of utilities instead of bounds on each individual utility. Again, we can easily recast the constraints in terms of agents' disruption, as follows:
    $\sum_{j=1}^n u_j(\sigma) + R \ge \sum_{j=1}^n u_j(\sigma_F) \Leftrightarrow \sum_{j=1}^n \delta_j \geq 0 $ and thus we obtain the following formulation of the problem: 
    \begin{align}
    &\max_{\sigma} \big\{ u_{n+1}(\sigma)\ :\  \sum_{j=1}^n \delta_j \geq 0 \big\}.\label{NA2}
    \end{align}
    In this case, the budget $R$ does not have to be distributed among the agents but can simply be added to the total utility.
    \item[Problem $1|\mbox{\it resch}, u_j \ge u_F|u_{n+1}$]. In this problem, we require that the value of the fairness in the rescheduled solution $\sigma$ has not worsened with respect to $\sigma_F$, i.e., all agents would reach a utility not less than $u_F = \min\{u_j(\sigma_F)\}$. 
    Recalling Equation~\eqref{def:Delta}, we note that in this case we are bounding agents' disruptions, as shown by the equivalence
    $u_j(\sigma) + \compj \ge u_F \Leftrightarrow \delta_j \geq u_F - u_j(\sigma_F)$. (Note that $u_F - u_j(\sigma_F) \le 0$ for all $j$.)
    Eventually we may write the problem as follows:
    \begin{align}
        &\max_{\sigma} \big\{ u_{n+1}(\sigma)\ :\  \delta_j \geq u_F - u_j(\sigma_F),\ j=1,\ldots,n\big\}. \label{NA3-SA}
    \end{align}
\end{description}
We start by showing that our first two versions of the problem, $1|\mbox{\it resch}, \delta_j\ge 0|u_{n+1}$ and $1|\mbox{\it resch}, \textstyle\sum\delta_j\ge 0|u_{n+1}$, are hard.


\begin{theorem}\label{th:NA1-hard}
Problem 
$1|\mbox{\it resch}, \delta_j\ge 0|u_{n+1}$
is weakly NP-hard, even for linear utility functions.  
\end{theorem}

\begin{proof}
We consider an instance $I_{P}$ of {\sc Partition} defined by $n$ integers $w_1,\ldots,w_{n}$ and let $B= \frac 1 2 \sum_{i=1}^n w_i$.
The problem asks for a subset $S$ of numbers summing up to $B$.

Now we construct an instance of the decision version of Problem~\eqref{NA1} with $n+2$ original jobs with linear utility functions, a new job $n+3$ and a budget $R$.
The large constant $M$ will be set later.
The $n+2$ jobs are scheduled by \gre\ yielding a schedule $\sigma_F$ with utilities $u_j(\sigma_F)$.
Then we ask, whether a new schedule $\sigma$ exists where job $n+3$ completes at time $2B$ (at the latest) while all the other jobs receive a utility and a payment $\compj$ fulfilling $u_j(\sigma) + \compj \geq u_j(\sigma_F)$. 
    
The data for the jobs are reported in the following table and we set $R=MB+3B^2$.
%
%
\begin{center}
\begin{tabular}{cccc}
\toprule
   $j$           & $p_j$  & $a_j$ & $b_j$ \\\midrule
   $1,\ldots, n$ &  $w_j$ & $w_j$ & $4B(M + 2B)$ \\
   $n+1$         & $M$    & $M^2$ & $\frac 5 2 M^3$ \\
   $n+2$         & $B$    & $\epsilon$ & $2 M^3 - 3M^2 B$ \\
   $n+3$         & $B$    & - & -   \\\bottomrule
\end{tabular}
\end{center}
We  choose $M$ large enough so that the schedule produced by \gre\ consists of jobs $1$ to $n$ in some unspecified order followed by jobs $n+1$ and $n+2$, see the upper schedule depicted in Figure~\ref{fig:Q2.1proof}.
To see that, we evaluate the utilities at time $P=2B+M+B$ which is the overall makespan.
We set $M:=5 B^2$ and obtain:
\begin{align*}
    &u_j(P) = 4B(M+2B) - a_j P \leq 20 B^3 + 8 B^2 \quad j=1,\ldots,n,\\
    &u_{n+1}(P) = \frac 5 2 M^3 - M^2(3B+M) = \frac 3 2 (5^3B^6) - 75 B^5,\\
    &u_{n+2}(P) =2 M^3 - 3 M^2 B - \epsilon(3B+M) = 2(5^3 B^6) - 75 B^5 - \epsilon(5 B^2 + 3 B).
\end{align*}
Hence $n+2$, which has the largest utility at $P$, is scheduled last. 

Before $n+2$, at time $2B+M$, we have 
$u_{n+1}(2B+M)=\frac 5 2 M^3 - M^2(2B+M) =\frac 3 2(5^3 B^6) - 50 B^5$, which is clearly larger than $u_j(2B+M)= 4B(M+2B) - a_j (2b+M) \leq 20 B^3 + 8 B^2$ for $j=1,\ldots,n$. 
So $n+1$ is scheduled directly before $n+2$.

The argument for placing the new job $n+3$ works as follows.
Inserting job $n+3$ to start not later than at time $B$ requires to postpone a subset $S$ of the jobs $1$ to $n$ with total processing time at least $B$. 
We will see that postponing job $n+1$ would cause excessive payments, while job $n+2$ can be deferred to the end of the schedule with minimal loss of utility.
Thus, we want to find a set $S$ of total length exactly $B$ and move these jobs in the gap left by job $n+2$.

\begin{figure}
    \centering
    \includegraphics[width=0.9\linewidth]{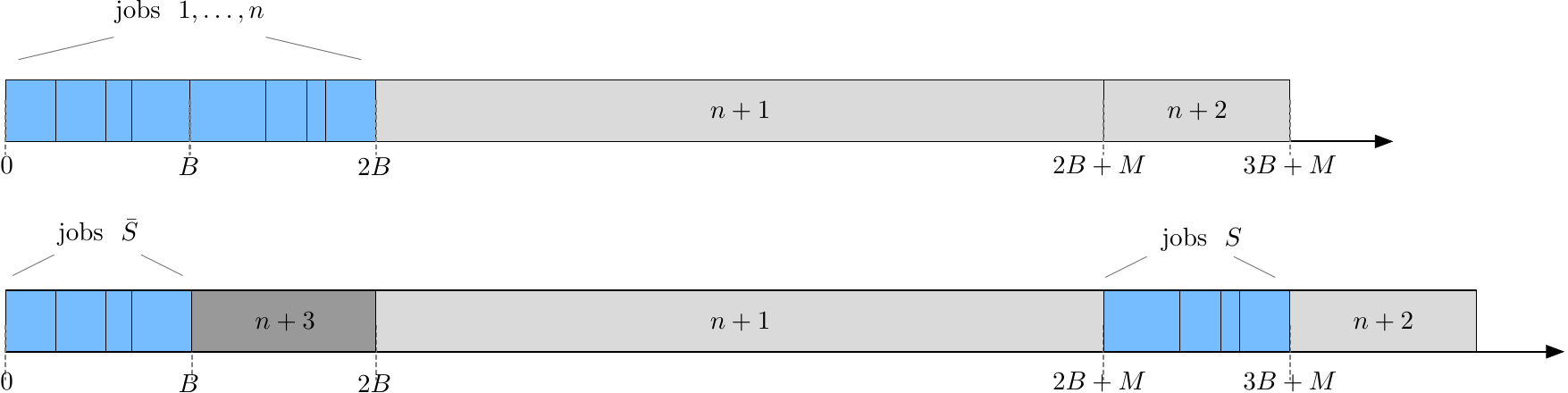}
    \caption{Original schedule $\sigma_F$ output by \gre\ (above) and the schedule after inserting job $n+3$ (below).}
    \label{fig:Q2.1proof}
\end{figure}

Clearly, jobs which complete in $\sigma$ not later than before will not decrease their utility. 
However, every job $j$ with a later completion time loses $a_j$ units of utility for each unit of time that it completes later than in the original schedule.

Increasing the completion time of job $n+1$ by only one unit of time would require a compensation payment of $M^2=25 B^4$, which exceeds the budget $R=5 B^3 + 3 B^2$.
Therefore, in $\sigma$ this job can not finish later than before.
On the other hand, job $n+2$ can be moved to the end of the new schedule with minimal loss of utility.

To place job $n+3$ somewhere in the interval $[0, 2B]$ it does not make sense to schedule jobs from $1,\ldots,n$ between jobs $n+3$ and $n+1$. 
Since $n+1$ starts at time $2B$ (or earlier) in schedule $\sigma$, we can always move $n+3$ to complete at the start of $n+1$. 
Therefore, some of the jobs from $1,\ldots,n$, denoted by subset $S$, have to be moved to start after jobs $n+3$ and $n+1$ with an earliest starting time of $M+B$. 
The other jobs in the set $\bar S = \{1,\ldots,n\}\setminus S$ will be scheduled at the beginning of $\sigma$ in the same sequence as before without reducing their utilities, see the lower schedule in Figure~\ref{fig:Q2.1proof}.

The payment required for moving some job $j\in S$ is at least its minimal loss of utility, namely $a_j (M+B - (2B-a_j)) = a_j(M-B+a_j)$ (earliest new starting time minus latest old starting time).

On the other hand, the maximum loss of utility for $j \in S$ is reached by moving the first job to the last position completing just before job $n+2$ at time $M+3B$ which is 
$a_j(M+3B-a_j)$.

Clearly, $\sum_{j\in S} p_j$ must be at least $B$ to accommodate job $n+3$.
For every subset $S$ with $\sum_{j\in S} p_j=B$, the total required payment is at most $\sum_{j\in S} a_j (M+3B -a_j) \leq \sum_{j\in S} a_j (M+3B) = MB + 3B^2 =R$, which is a feasible solution.
However, for every subset $S$ with $\sum_{j\in S} p_j\geq B+1$, the required payment is at least $\sum_{j\in S} a_j(M-B+a_j) \geq \sum_{j\in S} a_j(M-B) \geq (B+1)(M-B) = 5B^3 + 4 B^2 -B$ which is greater than $R$.
Thus, the completion time $2B$ can be reached for the new job $n+3$ if and only if there is a subset $S$ with $\sum_{j \in S}w_j=B$.
\end{proof}

\begin{theorem}\label{th:NA2-hard}
Problem 
$1|\mbox{\it resch}, \textstyle\sum\delta_j\ge 0|u_{n+1}$  is weakly NP-hard, even for linear utility functions.
\end{theorem}

\begin{proof}
To show NP-hardness of Problem 
$1|\mbox{\it resch}, \textstyle\sum\delta_j\ge 0|u_{n+1}$, we employ a scheduling problem discussed by~\cite{Struse16-14}, who introduced the problem of scheduling jobs on a single machine where a break of length $\Gamma$ has to be included in the schedule to allow some maintenance operation. 
The start time $\tau$ of the non-availability period can be chosen freely as long as it is not later than a given deadline $D_{MP}$.
For the objective of total weighted completion time the problem is denoted as $1|\tau \leq D_{MP}, \Gamma| \sum w_jC_j$
with the decision version asking for a schedule with $\sum w_jC_j \leq \bar W$. 
It is shown by~\cite{Struse16-14} 
that this problem is weakly NP-hard.

Given an instance $I_{MP}$ of the decision version of $1|\tau \leq D_{MP}, \Gamma| \sum w_jC_j \leq \bar W$, we construct an instance $I$ of the decision version of our problem 
with target value $C^T$ (for the
new agent $n+1$) as follows. Jobs $1$ to $n$ have the same processing times of the $n$ jobs in $I_{MP}$,
while $p_{n+1}=\Gamma$ and $C^T = D_{MP}+\Gamma$.
Linear utility functions are given with $a_j=w_j$ and arbitrary $b_j$.
For convenience, let $\bar f := u_{\sys}(\sigma_F) = \sum_{j=1}^n u_j(\sigma_F)$.
Finally, the budget in $I$ is set to $R:= \bar W + \bar f - \sum_j b_j$.
(Note that we may have $R<0$, but this can be accommodated 
representing a required increase of total utility.)

We will consider identical schedules of jobs $1,\ldots,n$ for both problems.
If $I$ 
is a {\sc yes}-instance with schedule $\sigma$ and completion times $C_j$, then $C_{n+1} \leq C^T$, which means that job $n+1$, corresponding to the maintenance operation of $I_{MP}$, starts at $C^T-p_{n+1} = D_{MP}$ at the latest, thus guaranteeing feasibility for $I_{MP}$. 
Moreover, the restriction on the utility compensation $\sum_{j=1}^n \delta_j\ge 0$ yields:
$$      \sum_{j=1}^n \big( u_j(\sigma) + \beta_j - u_j(\sigma_F) \big) \ge 0 \Leftrightarrow  \sum_{j=1}^n b_j - \sum_{j=1}^n a_j C_j + \sum_{j=1}^n  \beta_j - \bar f \ge 0.
$$
Since the budget $R$ is clearly fully exhausted, this is equivalent to
$
\sum_{j=1}^n b_j - \sum_{j=1}^n a_j C_j +  \bar W + \bar f - \sum_{j=1}^n b_j - \bar f \geq 0$
and therefore we have 
$\sum_{j=1}^n a_j C_j = \sum_{j=1}^n w_j C_j \leq \bar W$.
Thus, $I_{MP}$ 
is a  {\sc yes}-instance.
The equivalence of the inequalities between the two problems completes the reduction.
\end{proof}
%
%
Note that the same argument applies if the initial schedule is $\sigma_{\sys}$ instead of $\sigma_F$. In fact, the specific initial schedule plays no role in the reduction; what matters is only its total utility value, $\sum_j u_j(\sigma_F)$.
As a consequence, we can state that 
     $1|\mbox{\it resch}, \textstyle\sum\delta_j\ge 0|u_{n+1}$ 
     remains weakly NP-hard, even when the initial solution is the system-optimal schedule $\sigma_{\sys}$ and the utility functions are linear.

We now turn to the last of the three versions of our single-agent problem---namely, 
$1|\textit{resch},\, u_j \ge u_F|u_{n+1}$---and show that this one can be solved efficiently.

\begin{theorem}\label{th:NA3-poly}
Problem~\eqref{NA3-SA}  $1|\mbox{\it resch}, u_j \ge u_F|u_{n+1}$ can be solved in $O(n^2 log P)$ time.  
\end{theorem}


Also for Problem~\eqref{NA3-SA}, starting from an initial system-optimal solution $\sigma_{\sys}$ and hence replacing the lower bound $u_F$ by $u_{\sys}$, does non change the validity of the procedure described in the proof of Theorem~\ref{th:NA3-poly}, which works for an arbitrary lower bound on the individual utility values. Therefore we may conclude that
    $1|\mbox{\it resch}, u_j \ge u_{\sys}|u_{n+1}$, i.e., Problem~\eqref{NA3-SA} with initial solution $\sigma_{\sys}$ and linear utility functions, is weakly NP-hard.

We close this section with a few additional observations about the latter result.
\begin{itemize}
[noitemsep,nolistsep]
    \item If no compensations are intended for agents $1,\ldots, n$, i.e., $R=0$, Problem~\eqref{NA3-SA} can be solved (again in $O(n^2)\log (\sum_{j=1}^{n+1}p_j)$ time), with a simplified procedure.
As in the proof of Theorem~\ref{NA3-SA}, we perform a binary search over the completion time $C_{n+1}$ of job $n+1$.
For each target value $C^T$ in the binary search we obtain a corresponding due date $d^T_{n+1}$ of job $n+1$ (while all other jobs have no due date). 
Thus, the decision problem arising in each iteration can be solved in $O(n^2)$ time by Theorem~\ref{th:no-latejobs}.
If the maximum minimum utility value of the resulting schedule is smaller than $u_F$, we have to consider a larger target value, otherwise we can test a smaller target value.
 \item Clearly, the process described in the proof of Theorem~\ref{th:NA3-poly} could be recast to the case in which we want to maximize the utility of one job $i$ among the original jobs $1,\ldots, n$. 
 In this case the binary search is restricted in the interval between $p_i$ and the completion time $C_i(\sigma_F)$ determined by {\gre}. 
\item It is also possible to extend the procedure by considering more than one job which should have a utility as high as possible (and thus a minimal completion time). For instance, we may maximize the utilities of an ordered list of jobs $j_1,\ldots,j_k$, $k\le n$, in a lexicographic order.
After determining the earliest completion time $C^*_{j_1}$ for the first job $j_1$ in the list (by the binary search process sketched above) we fix this value as due date $d_{j_1}:=C^*_{j_1}$ and iterate with minimizing the completion time of the next job in the list. The implied decision problems can still be solved by Theorem~\ref{th:no-latejobs}.
\end{itemize}


\section{Enforcing a target sequence in a bi-level setting}\label{sec:target}

Like in a leader-follower scenario, an external party (the leader) might be interested to enforce a certain schedule $\sigma_T$, e.g., minimizing sum of completion times to optimize throughput, or maximizing total utility.
In such a bi-level setting, the follower would still be solving  Problem $1 | u_j = b_j - a_j C_j | u_{\min}$ to reach a fair solution $\sigma_F$ for all jobs, but the leader influences the outcome of this process by modifying the utility functions of jobs.

Suppose first that---in analogy to Section~\ref{sec:bound-moda}---the leader is allowed to either modify the intercept values by increasing $b_j$ by adding $\db_j$ to it, or to decrease the slopes by subtracting $\da_j$ from $a_j$. Hence, it is required that $\db_j\geq 0$, resp.\ $0 \leq \da_j \leq a_j$.
In this way, the leader influences the sequence determined by the follower for the modified utility functions. Our problem asks for minimizing $\sum_j \db_j$, resp.\ $\sum_j \da_j$, such that the fair solution $\sigma_F$ generated by the follower, i.e., the maximization of the minimum utility, corresponds to the desired target sequence $\sigma_T$.
This setting is not enough for a meaningful problem definition, since the maximization of minimum utility leaves a wide range of different sequences with identical minimum utility. The leader would have no possibility to influence the sequence chosen by the follower among all the solutions tied for fairness. Thus, we will make the action taken by the follower more precise and assume that it uses \gre\ for determining the fair solution.

According to the Graham-like notation introduced above we are denoting the above two problems respectively      
$1 | \sigma_F = \sigma_T, u_j = \b_j - a_j C_j | \sum \b_j$ and 
$1 | \sigma_F = \sigma_T, u_j = b_j - \a_j C_j | \sum \a_j$.

These problems can be solved as follows.
 Consider jobs numbered in increasing order of completion times as given by the target sequence $\sigma_T$, i.e.,
$$
C_j(\sigma_T) \ge C_{j-1}(\sigma_T), \quad j = 2,\ldots, n.
$$
The key idea is that, for \gre\ to schedule job $j \ge 2$ at time $C_j$, it suffices to choose the parameters $\db_j$ or $\da_j$ so that $u_j(C_j) \geq u_\ell(C_j)$ for all $\ell = 1,\ldots,j-1$.
All jobs $\ell$ with earlier completion times $C_\ell < C_j$ will not change their values anymore.
Conversely, all jobs with larger completion times than $C_j$ will take the new utility value of $j$ into account when their values are determined.
At any point in time $T$, \gre\ considers only the jobs to be placed before $T$. 
Thus, the later jobs in the sequence do not have to be considered in the choice of $\db_j$ resp.\ $\da_j$.

Formally, we start by choosing $\db_1 := 0$ (resp., $\da_1=0$) and, for $j\ge 2$, we increase the $b_j$ values by setting
    $\db_j := \max\{0, \b_\ell-a_\ell C_j - b_j + a_j C_j, \ell=1,\ldots, j-1\}$.
Note that any desired target sequence can be imposed by increasing the $b_j$ values.

For the decrease of $a_j$ we set
    $\da_j:= \max\{a_j-\a_\ell + \frac{b_\ell-b_j}{C_j}, \ell=1,\ldots, j-1\}$.
If this choice yields $\da_j > a_j$, this would yield an increasing utility function violating the general assumption of nonincreasing utilities.
Thus, not every desired target sequence can be obtained by reducing the $a_j$ values (although this changes when an increase of $a_j$ is also allowed, see below).

If the follower aims at the maximization of global utility,
it will schedule the jobs in non-increasing order of $a_j/p_j$.
Obviously, the modification of intercepts $b_j$ is useless for this case.
As for the slopes $a_j$, they are modified such that this ordering corresponds to the schedule $\sigma_T$ desired by the leader.

In the resulting problem 
$1 | \sigma_{\sys} = \sigma_T, u_j = b_j - \a_j C_j | \sum \da_j$
we have to choose $\da_j$ such that
$\a_j/p_j \leq \a_\ell/p_\ell$ for all $\ell < j$.
This yields 
$$\da_j := \max\{0, a_j - p_j \cdot\min_{1\leq\ell < j}\{\frac{\a_\ell}{p_\ell}\}\}.$$

We can also extend the model and allow the modification of $b_j$ in both direction, meaning that $\db_j$ could be positive or negative, where the objective function asks for the minimization of the sum of absolute values of all $\db_j$.
Note that in this case it is not clear whether it would be beneficial from a global perspective to increase the intercept of the currently considered job $j$, or rather decrease the intercepts of all the other jobs, which so far prevent $j$ from being chosen by \gre. 
Such a decrease of other jobs might be beneficial for future decisions of the algorithm.
A similar extension can be considered for the slope where we allow negative $\da_j$, i.e., an increase of the negative slope.
Clearly, the above approach does not work anymore for this general case and we are not aware of a combinatorial algorithm for the resulting optimization problem.
However, we can give an LP-formulation, which implies that the problem remains polynomially solvable.
For convenience, we denote by $P_j:=p_1+\ldots + p_j$ the desired completion time of job $j$.
\begin{eqnarray}
\min \sum_{j\in N} |\db_j|&&\\
\mbox{s.t. } \db_j - \db_i &\geq& b_i-b_j - (a_i-a_j)P_j\quad 1\leq i<j \leq n\label{eq:utility-constraint}\\
\db_j &\in& \Real \quad \quad j=1,\ldots,n
\end{eqnarray}
Conditions (\ref{eq:utility-constraint}) are derived from the selection criterion of \gre\ for job $j$ at time $P_j$, i.e., $\db_j + u_j(P_j) \geq \db_i + u_i(P_j)$.
The absolute values in the objective function can be handled by standard methods of LP-modeling.

An analogous model works for the modification of the slopes.
\begin{eqnarray}
\min \sum_{j\in N} |\da_j|&&\label{eq:absolute slopes}\\
\mbox{s.t. } (\da_j - \da_i) P_j &\geq& b_i-b_j - (a_i-a_j)P_j\quad 1\leq i<j \leq n\label{eq:utility-constraint slope}\\
\da_j &\leq& a_j,\quad \da_j \in \Real \quad j=1,\ldots,n
\end{eqnarray}
Note that contrary to the case $\da_j \geq 0$, in the general case every target sequence can be attained by a suitable choice of $\da_j$.

For the corresponding case of $1 | \sigma_T:=\sigma_{\sys}, u_j = b_j - \a_j C_j | \sum \da_j$ with possibly negative $\da_j$, we replace \eqref{eq:utility-constraint slope}
by
\begin{equation}\label{eq:utility-constraint slope sum}
   \da_j -  \frac{p_j}{p_i} \da_i \leq a_j - \frac{p_j}{p_i}a_i \quad 1\leq i<j \leq n.
\end{equation}

\def\R{\mathcal{R}}
For this case we can also derive a simple combinatorial algorithm based on dynamic programming.
It is based on the following observation.
When we choose $\da_j$ to reach an ordering of jobs in non-increasing order of $\a_j/p_j$, the new slopes will always be set in such a way that the new ratios $\a_j/p_j$ attain one of the $n$ original values $a_i/p_i$.
Let $\R:=\{a_i/p_i \mid 1\leq i \leq n \}$ be this set of all relevant ratio values.
Otherwise, let $J_r$ be the set of all jobs $j$ with $\a_j/p_j=r$ and $r \not\in\R$.
A marginal increase of $r$ for all jobs in $J_r$ to $r+\epsilon$ would allow a decrease of the slope reduction for jobs with positive $\da_j$ by $\epsilon p_j$ (say these are jobs $J_r^+$), but require a further reduction of all negative $\da_j$ by $\epsilon p_j$ (for jobs $J_r^-$).
The same applies with exchanged signs for a marginal decrease of $r$ to $r-\epsilon$.
Therefore, the total increase of the objective function \eqref{eq:absolute slopes}
is $\sum_{j \in J_r^-} \epsilon p_j - \sum_{j \in J_r^+} \epsilon p_j$ for reaching $r+\epsilon$, while it is $ - \sum_{j \in J_r^-} \epsilon p_j + \sum_{j \in J_r^+} \epsilon p_j$ for reaching $r-\epsilon$.
Clearly, one of the two contributions has to be negative and thus improves the objective function.
The value of $\epsilon$ in this direction can be extended until the nearest value in $R$ is reached.

We introduce a dynamic programming array representing the cost of setting the ratio of job $j$ to $r \in \R$.
Formally we define:
\begin{quote}
$c_j[r] = $ cost of setting $\a_j/p_j=r$ plus the minimum cost such that the ratios $\a_i/p_i$ of all jobs $i=j+1,\ldots,n$ are sorted in non-increasing order.
\end{quote}
The computation of $c_j[r]$ is performed in decreasing order of $j$, i.e.\ $j=n, n-1,\ldots,1$ for every $r\in \R$.
To reach $\a_j/p_j=(a_j-\da_j)/p_j=r$, we have costs of $\da_j=a_j - p_j\cdot r$.
Including the cost for the correct ordering of jobs $j+1, \ldots, n$, which is already collected in $c_{j+1}[r]$, we get
\begin{equation}
    c_j[r] := | a_j - p_j\cdot r | + \min_{\rho \in \R}\{ c_{j+1}[\rho] : \rho\leq r\}.
\end{equation}
Going through the values of $r$ in increasing order, we can compute $c_j[r]$ in linear time so that the total running time for solving
$1 | \sigma_T:=\sigma_{\sys}, u_j = b_j - \a_j C_j | \sum \da_j$ with arbitrary $\da_j$ is  $O(n^2)$.

\section{Conclusions}\label{sect:concl}

In this work, we have investigated a broad family of multi‑agent scheduling problems with the aim of reaching a fair solution (in a maximin sense) for agents which may differ in their utility functions, temporal constraints, and degrees of flexibility in adjusting their parameters. 
Our analysis clarifies how release dates, due dates, and agent‑specific utility models shape the underlying difficulty of the problem, and how controlled modifications of the input parameters can be leveraged to attain desirable schedules.
Building on these insights, we have developed exact approaches that enable the computation of fair schedules, and -- whenever this seems meaningful for the considered problem setting -- we also treated the efficient version, where the sum of utilities is maximized. 
Furthermore, our exploration of adjustable utility functions and leader–follower scenarios highlights how external interventions, such as budget‑bounded modifications of utilities, can influence the resulting schedules, revealing connections between fairness, incentives, and strategic control in multi‑agent environments.

It is worth to mention here that our study extends to more general notions of fairness, such as the \emph{Kalai-Smorodinsky} fair solution, in which  the minimum of the \emph{normalized utility} across the agents is maximized \citep{bib:ks1975}. 
More precisely, looking at the range 
$I_j = \max_{\sigma\in \S} \{u_j(\sigma)\} - \min_{\sigma\in \S} \{u_j(\sigma)\}$ 
of utility values that $j$ may obtain among any 
semi-active schedule, the normalized utility $\bar u_j$ of agent $j$ is defined as the ratio $\bar u_j(\sigma) = {u_j}/{I_j}.$
For instance, recalling Definition~\ref{eq:fairsoldef}, in the case of linear utility functions and no release dates, if $P=\sum_j p_j$, then
$I_j = a_j(P - p_j)$ and hence the normalized utilities for all $j$ are known for any given schedule.
If this range can be computed in advance for every agent,
any algorithm that computes $u_F$ can also be applied to obtain a Kalai-Smorodinsky fair solution (although the resulting schedules may be different).

For future research one can consider the case of $m$ parallel machines, possibly with different specifications.
Looking at the special case with $m=2$ parallel identical machines, one can easily show that our problem is weakly NP-hard, even if all jobs have the same linear utility function.
In a similar flavor as in Section~\ref{sec:release-fair}, a pseudopolynomial algorithm based on dynamic program can be devised for the latter problem.

\bibliographystyle{apalike}

\bibliography{fair_sched}

\newpage

\section*{Appendix}

\subsection*{Proof of Theorem~\ref{th:greedy}}
\begin{quote}{\it
If, for all $j\in J$,
$u_j$ is any non-increasing function of the completion time $C_j$ of job $j$, then 
\gre\ computes an optimal solution of \ $1|u_j|u_{\min}$ in $O(n^2)$ time.
}\end{quote}
\par\noindent\medskip
For a given schedule $\sigma$, let $z(\sigma) = \min_{j\in J} u_j(C_j)$ indicate its corresponding solution value. 
Let $\sigma_F$ be an optimal schedule in which $i\in J$ is the last job: 
Since we may assume w.l.o.g. that  $\sigma_F$ has no idle times, $i$ completes at $C_i(\sigma_F) = P =\sum_{j\in J} p_j$. 

Let $\jstar\in J$ be the job selected by the algorithm, i.e., $\jstar = \arg \max_{h\in J} \{u_h(P)\}$ and assume, by contradiction, that $u_i(P) < u_{\jstar}(P)$.
Consider a schedule $\tilde\sigma$ obtained from $\sigma_F$ by just postponing job $\jstar$ and sequencing it as the last job, immediately after $i$.
In $\tilde\sigma$ all jobs, but $\jstar$, have not increased their completion times and hence, recalling our objective function definition: 
$z(\sigma_F) \le u_h(C_h(\sigma_F) \le u_h(C_h(\tilde\sigma))$ for all $h\in J\setminus\{\jstar\}$.
On the other hand  $z(\sigma_F) \le u_i(C_i(\sigma_F)) = u_i(P) < u_{\jstar}(P) = u_{\jstar}(C_{\jstar}(\tilde\sigma))$.

In conclusion $z(\sigma_F) \le \min_{h\in J}\{u_h(C_h(\tilde\sigma))\} = z(\tilde\sigma)$, so $\tilde\sigma$ is optimal as well. This argument can be recursively iterated at time $P-p_{\jstar}$, showing the correctness of \gre. \qed

\subsection*{Proof of Theorem~\ref{th:release_strongNPH}}

\begin{quote}{\it
Problem $1|r_j,u_j|u_{\min}$ 
is strongly NP-hard, even for linear utility functions.
}\end{quote}

\par\noindent\medskip
Consider an instance $I_L$ of the scheduling problem with due dates and release dates where maximum lateness $L_{\max}$ is minimized, i.e., $1|r_j|L_{\max}$, which was shown to be strongly NP-hard by~\cite{bib:lenstra1977}.
We can construct an instance of our problem with the same release dates and utility functions $u_j=d_j-C_j + M$, where $M > 0$ is large enough to guarantee that $u_j\ge 0$ for any reasonable value of $C_j$.
Considering the solution of our problem as a solution for $I_L$, it is easy to see that the lateness $L_j:=C_j-d_j= -u_j + M$.
If there is a solution with $u_F \geq K + M$, then for $I_L$ there is a solution with $L_{\max} \leq -K$, and vice versa.
\qed

\subsection*{Proof of Theorem~\ref{th:release_binary}}
\begin{quote}{\it
Problem $1|r_{j},u_j|u_{\min}$ 
is weakly NP-hard, even if the  utility functions are linear and only 
one release date is nonzero.
}\end{quote}

\par\noindent\medskip
For only one release date $r_{n+1}\neq 0$, $1|r_j,u_j|u_{\min}$\ can be reduced from {\sc Partition}.
Let the instance $I_{P}$ of {\sc Partition} be defined by $n$ integers $w_1,\ldots,w_{n}$ and let $B= \frac 1 2 \sum_{i=1}^n w_i$.
The problem asks for a subset $S$ of numbers summing up to $B$.
The corresponding instance $I$ of our scheduling problem is the following. We have $n+1$ jobs, where the first $n$ jobs correspond to the items in $I_{P}$. The data are:

\[ p_{j} = \left\{ \begin{array}{ll}
    w_j & \mbox{for } j=1,2, \ldots, n  \\
    B & \mbox{for } j=n+1
\end{array}
\right .\]

\[ r_{j} = \left\{ \begin{array}{ll}
    0 & \mbox{for } j=1,2, \ldots, n  \\
    B & \mbox{for } j=n+1
\end{array}
\right .\]

\[ u_{j}(C_j) = \left\{ \begin{array}{ll}
    2B - \frac{1}{3}C_j & \mbox{for } j=1,2, \ldots, n  \\
   5B -2 C_j  & j=n+1
\end{array}
\right .\]

If $I_{P}$ is a {\sc yes}-instance, then in $I$ one can schedule the jobs of $S$ between time $0$ and $B$, then start job $n+1$ at time $B$ and append the remaining jobs $\{1,\ldots,n\}\setminus S$ between time $2B$ and $3B$.
This yields utility $5B - 2\cdot 2B=B$ for job $n+1$ and utility $2B-\frac 1 3 3B=B$ for the job scheduled last among $1,\ldots,n$, i.e.\ the job reaching the lowest utility, which gives $u_F=B$.

If $I_{P}$ is a {\sc no}-instance, one can either schedule job $n+1$ at time $B$, reaching again utility $B$.
Since the total length of jobs scheduled before time $B$ must be strictly less than $B$ ({\sc no}-instance), the remaining jobs to be scheduled after job $n+1$ from time $2B$ must have total length at least $B+1$.
Thus, the utility of the last job is at most $2B-\frac 1 3( 2B+B+1)\leq B-\frac 1 3$.
Otherwise, if job $n+1$ is scheduled later than $B$, its utility is at most $5B - 2(B+1+B)=B-2$.
In both cases there is $u_F<B$, which means that the optimal solution of $I$ would decide {\sc Partition}.
\qed


\subsection*{Proof of Theorem~\ref{th:NA3-poly}}
\begin{quote}{\it
Problem~\eqref{NA3-SA}  $1|\mbox{\it resch}, u_j \ge u_F|u_{n+1}$ can be solved in $O(n^2 log P)$ time.  
}\end{quote}

\par\noindent\medskip
Problem $1|\mbox{\it resch}, u_j \ge u_F|u_{n+1}$ can be solved by combining two results of this paper (namely, Theorem~\ref{alg:greedy} and the procedure given in Section~\ref{sec:bound-modb}). 
First of all, we perform a binary search for the minimum completion time 
$C_{n+1}$ of job $n+1$, taking $O(\log P)$ iterations.
Clearly, $C_{n+1}$ must be between $p_{n+1}$ and $P = \sum_{j=1}^{n+1} p_j$.
Each target value $C_{n+1}^T$ of the binary search implies
a utility function $u_{n+1}(C)= M$ for $C \leq C_{n+1}^T$ and $u_{n+1}(C)=0$ otherwise.
Then we can run \gre\ to obtain a feasible schedule $\sigma^T$ with maximum minimum utility $u_{\min}(\sigma^T)$ in $O(n^2)$ time.
For this schedule we apply the linear time procedure given in Section~\ref{sec:bound-modb} to increase the utility functions by constant shifts $\db_j$ representing the payments $\compj$ with a given budget bound $B=R$ (which applies also for arbitrary utility functions).
Note that during this increase the (auxiliary) utility function $u_{n+1}(C)$ of job $n+1$ will never be modified.
If the resulting improved maximum minimum is at least $u_F$ we can continue the binary search with a target value greater than $C^T$, otherwise a smaller target value is chosen.
\qed

\end{document}